\documentclass[runningheads]{llncs}

\usepackage[T1]{fontenc}
\usepackage{lmodern}
\usepackage{microtype}
\usepackage{amsmath,amssymb,mathtools,amsfonts}
\usepackage{graphicx}
\usepackage{float}
\usepackage{xcolor}
\usepackage{url}
\usepackage{array,tabularx}
\usepackage[mathlines]{lineno}
\usepackage{hyperref}
\hypersetup{hidelinks}
\usepackage{tikz}
\usepackage{arydshln}
\usepackage[scale=0.75]{geometry}
\usetikzlibrary{arrows.meta,calc,intersections,positioning,decorations.markings}

\spnewtheorem{observation}{Observation}{\bfseries}{\itshape}
\spnewtheorem{result}{Fact}{\bfseries}{\itshape}

\newcommand{\area}{\operatorname{area}}

\begin{document}
\title{Largest-Area Convex Quadrilateral in a 1.5D Terrain}
\titlerunning{Largest-Area Convex Quadrilateral in a 1.5D Terrain}
\author{Nandana Ghosh\inst{1} \and Ankush Acharyya\inst{1}}
\authorrunning{N. Ghosh and A. Acharyya}
\institute{National Institute of Technology Durgapur, India
\\ \email{\{ng.23cs1109,~aacharyya.cse\}@nitdgp.ac.in}}

\maketitle

\begin{abstract}
A $1.5$D terrain is a simple polygon bounded by a horizontal base and an $x$-monotone upper chain. We study the problem of finding a largest-area convex quadrilateral contained in an $n$-vertex terrain. We maximize area over the closure of the feasible nondegenerate quadrilaterals, allowing a triangular boundary optimum when necessary. Assuming that no three terrain vertices are collinear, we give a deterministic exact algorithm running in $O(n^2\log n)$ time and using $O(n)$ working space in the algebraic real-RAM. Among all optimum solutions,
the algorithm returns a nondegenerate quadrilateral whenever one exists; otherwise, it returns a maximum-area terrain triangle.

We also prove that a maximum-area axis-parallel rectangle contained in the terrain yields a tight $\frac12$-approximation and can be computed in $O(n\log n)$ time.
\end{abstract}

\section{Introduction}\label{sec:introduction}

\paragraph{Problem and motivation.}
The classical potato-peeling problem asks for a maximum-area convex subset
of a simple polygon~\cite{chang1986polynomial,goodman1981largest}. For an
arbitrary polygon, exact algorithms have high polynomial running times, so
it is natural to ask whether substantially faster algorithms are possible
for geometrically structured polygon classes.

A $1.5$D terrain is one such class. Its lower boundary is a horizontal base
and its upper boundary is an $x$-monotone polygonal chain. Thus every
vertical line intersects the terrain in a single interval, and moving a
feasible point vertically toward the base preserves feasibility. The
terrain also has a particularly simple shortest-path structure: shortest
paths rooted at the two base endpoints are lower convex chains, and their
last edges admit linear-size families of forward prolongations with local
noncrossing order. These properties are strong enough to support the
$O(n\log n)$ exact algorithm of Cabello et al.~\cite{CabelloDDM25} for a
largest-area terrain triangle, improving the earlier $O(n^2)$ bound of
Das et al.~\cite{DasDM21}.

The quadrilateral problem is the next natural case, but it introduces
phenomena that do not arise for triangles. A grounded quadrilateral has
three nonbase sides, and the two base-incident sides need not be determined
solely by extremal terrain contacts. A side can be stationary at a
\emph{balanced} position, where an interior terrain contact becomes its
midpoint, or a locally feasible rotation can be stopped by the terrain near
an upper endpoint. These additional local optima are the main reason that
the triangle algorithm does not extend directly.

We therefore consider the following problem. Given an $n$-vertex terrain,
find a convex quadrilateral of maximum possible area. To treat limiting
configurations cleanly, we optimize over the closure of the family of
nondegenerate feasible quadrilaterals. Under this convention, a triangular
limit is part of the feasible closure; if a nondegenerate quadrilateral
attains the same optimum value, we prefer it.

\paragraph{Related work.}
Goodman introduced the potato-peeling problem~\cite{goodman1981largest},
and Chang and Yap gave an exact $O(n^7)$-time algorithm for arbitrary simple
polygons~\cite{chang1986polynomial}. Cabello et al. later obtained a
randomized $(1-\varepsilon)$-approximation in
$O\!\left(n(\log^2 n+\varepsilon^{-3}\log n+\varepsilon^{-4})\right)$
time~\cite{CabelloCKSV17}. Closely related geometric optimization problems
for prescribed shapes include triangles via shortest-path
structures~\cite{MelissaratosS92}, axis-parallel rectangles in
$O(n\log n)$ time~\cite{BolandU01a}, and arbitrarily oriented rectangles in
$O(n^3\log n)$ time~\cite{ChoiLA19}.

For convex polygons, the maximum-area quadrilateral problem has a different
character. Keikha et al. showed that the Dobkin--Snyder algorithm is not
always correct~\cite{Keikha_2017_abs}, and Rote subsequently gave a
deterministic linear-time algorithm~\cite{rote2019largest}. These results
rely on convexity of the containing polygon and do not address the
obstructions created by the reflex vertices of a terrain.

For $1.5$D terrains, Das et al.~\cite{DasDM21} and Cabello et
al.~\cite{CabelloDDM25} gave exact algorithms for the largest-area triangle,
with the latter running in $O(n\log n)$ time. For every fixed $k>3$, Keikha
gave FPTASs for maximum-area and maximum-perimeter convex polygons with at
most $k$ vertices contained in a terrain, together with exact algorithms
for the longest segment and the largest-perimeter
triangle~\cite{Keikha_abs}. To the best of our knowledge, the present result
is the first exact algorithm beyond $k=3$ for an unrestricted maximum-area
convex terrain $k$-gon; here ``unrestricted'' means that no geometric shape
is prescribed beyond convexity and the number of vertices.

\paragraph{Our results and techniques.}
Our main result is a deterministic exact algorithm that runs in
$O(n^2\log n)$ time and uses $O(n)$ working space under a standard
general-position assumption. The first step is structural. We show that an
optimum can be chosen with one side on the terrain base. For such a
grounded quadrilateral, local perturbations lead to three algorithmic
possibilities for each base-incident side: it is represented by a
precomputed terrain carrier, it is balanced at a terrain vertex, or it is
incident to an exposed upper corner.

The carrier case is the closest to the terrain-triangle setting, but a
quadrilateral still has a continuously moving top support. For a fixed pair
of side carriers, we show that the relevant top supports are exactly lower
supports of a changing lower hull. A monotone hull sweep processes all right
carriers for one fixed left carrier, with an $O(\log n)$ search for each
pair. Over all carrier pairs this gives the $O(n^2\log n)$ carrier branch.

Balanced sides require a different reduction. For every possible balanced
centre we derive the exact interval of feasible upper endpoints. Inside
each sweep component, the possible top sides are represented by a linear
family of auxiliary forward carriers. Their intersections with the original
right carriers admit an anti-Monge inequality. A hereditary segment-tree
decomposition converts the relevant finite-segment intersections into
totally monotone matrices, which are searched by SMAWK
~\cite{BurkardKR96,AggarwalKMSW87}. This gives
$O(n\log n)$ time for one balanced centre and $O(n^2\log n)$ in total.

The remaining difficulty is caused by exposed upper corners. We construct
linear-size maps describing the extreme base points visible from a moving
point on the upper terrain chain. With one exposed corner, kinetic tangent
and visibility sweeps reduce the search to $O(n)$ one-parameter event
intervals for each fixed carrier or balanced centre. With two exposed
corners, the same information yields an $O(n^2)$ decomposition into
constant-complexity semialgebraic cells, each optimized exactly by
fixed-degree algebraic methods. Horizontal-top configurations that contain
a balanced side are handled separately in $O(n^2)$ time.

Finally, we prove a simple approximation theorem with a different flavor.
Every grounded optimum with base angles at most $90^\circ$ contains an
axis-parallel rectangle of at least half its area. Combined with the
$O(n\log n)$ algorithm of Boland and Urrutia~\cite{BolandU01a}, this gives
a tight $\frac12$-approximation in $O(n\log n)$ time.

\paragraph{Preliminaries.}
Let $V(\mathfrak T)=\{v_1,\ldots,v_n\}$ be the terrain vertices, indexed in
increasing $x$-order. The segment $\overline{v_1v_n}$ is the terrain base,
with endpoints $B_\ell=v_1$ and $B_r=v_n$, and $\Gamma$ denotes the upper
terrain chain. We place the base on the $x$-axis $\Pi$ and assume that the
terrain lies above it. Since $\Gamma$ is $x$-monotone, it is the graph of a
continuous piecewise-linear function; we denote its height at abscissa $x$
by $T(x)$. Thus the vertical section of the terrain at $x$ is exactly
$\{x\}\times[0,T(x)]$.

For $p,q\in\Gamma$, we write $p\prec q$ when $p_x<q_x$ and
$p\preceq q$ when $p_x\le q_x$. A terrain vertex is \emph{convex} or
\emph{reflex} according as its interior angle is smaller or larger than
$180^\circ$. For a point $p$, its coordinates are denoted by $p_x,p_y$; for
points $a,b$, $|ab|$ denotes the length of $\overline{ab}$. For a planar
region $X$, we write $\area(X)$ for its area. For
$S\subseteq\mathbb R^2$, $\operatorname{conv}(S)$ denotes its convex hull
and $\operatorname{LH}(S)$ its lower hull, ordered from left to right.

For a positive-length segment $\rho$, $\operatorname{line}(\rho)$ denotes
its supporting line. If $g$ is nonvertical, $g(x)$ denotes the height of
its supporting line at abscissa $x$; we use the same notation when $g$ is
itself a segment. A \emph{lower supporting line} of a point set is a line
that meets the set and has the entire set on or above it. If such a
supporting line is required to pass through a specified point, we call it a
\emph{lower tangent} from that point. Unless stated otherwise, slope means
$\Delta y/\Delta x$.

We use the orientation predicate
\[
\operatorname{orient}(p,q,r)
=
\det
\begin{pmatrix}
q_x-p_x & q_y-p_y\\
r_x-p_x & r_y-p_y
\end{pmatrix},
\]
so clockwise triples have negative orientation. Four supporting lines are
said to be in \emph{convex support order} when their successive
intersections form a clockwise quadrilateral. Equivalently, the four
consecutive orientation determinants are negative. Weak inequalities are
used only for closed boundary cases. Unless stated otherwise,
quadrilaterals are convex, nondegenerate, and listed clockwise.

For a level $h$, a \emph{level component} is a connected component of
$\{x:T(x)\ge h\}$. Sweep arguments need a slightly finer decomposition.
Within a level component, we split at isolated points and maximal segments
where $T(x)=h$. The closures of the maximal intervals on which $T(x)>h$ are
called \emph{sweep components}; the splitting points and horizontal
height-$h$ segments are retained as closed level events. The region between
a sweep component and its corresponding upper-chain subarc, completed by
the two vertical boundary segments, is a topological disk. A
\emph{cross-cut} of this disk is a simple arc joining two boundary points
and otherwise lying in the interior.

Throughout the exact algorithm we assume \emph{general position}: no three
terrain vertices are collinear. This assumption still allows two terrain
vertices to have the same height and allows an optimal side to coincide
with a terrain edge.

\paragraph{Closed feasible family and output convention.}
Let $\mathcal Q(\mathfrak T)\subseteq\mathfrak T^4$ be the family of
nondegenerate convex quadrilaterals contained in $\mathfrak T$, represented
by their clockwise ordered vertex quadruples. We optimize over the closure
$\overline{\mathcal Q(\mathfrak T)}$, where the closure is taken in
$\mathfrak T^4$. Since $\mathfrak T$ is compact,
$\overline{\mathcal Q(\mathfrak T)}$ is compact as well, and continuity of
area guarantees that the maximum is attained.

The use of the closure records the natural triangular boundary cases of the
quadrilateral problem. A positive-area degenerate member can be represented
as a triangle with a repeated or collinear fourth vertex. Conversely, every
feasible positive-area terrain triangle belongs to
$\overline{\mathcal Q(\mathfrak T)}$, because one of its vertices can be
truncated to obtain nondegenerate feasible quadrilaterals converging to it.
Let $\Delta^*$ denote a maximum-area terrain triangle, which can be
computed in $O(n\log n)$ time by Cabello et al.~\cite{CabelloDDM25}.
Among all solutions attaining the optimum value, our output convention
prefers a nondegenerate quadrilateral whenever one exists; otherwise the
algorithm returns $\Delta^*$. We write $Q^*$ for an optimum under this
convention, identifying a degenerate optimum with its triangular convex
hull. A feasible polygon is \emph{grounded} if one of its sides lies on the
terrain base $\overline{B_\ell B_r}$.

\paragraph{Computational model.}
We work in an algebraic real-RAM with exact arithmetic and comparisons.
Roots of univariate polynomials of fixed degree can be isolated, evaluated,
and compared in constant time. The algebraic optimization problems that
arise later have only a constant number of variables, constraints, and
degrees. The following fact makes the resulting $O(1)$-per-cell
optimization statement explicit.

\begin{result}\label{res:algebraic-cell}
Fix constants $r$ and $d$. Let $S\subset\mathbb R^2$ be a nonempty compact set defined by at most $r$ polynomial weak inequalities of degree at most $d$, and let $f(x,y)=\frac{P(x,y)}{Q(x,y)}$, with $\deg P,\deg Q\le d$, where $Q$ has fixed nonzero sign on $S$. Then $\max_S f$ and a point attaining it can be computed in $O(1)$ time in the algebraic
model above. Moreover, if a preferred subset of $S$ is described by a constant number of strict polynomial inequalities of bounded degree, one can decide in the same bound whether the maximum level set contains a preferred point, return one if it does.
\end{result}

\begin{proof}
Since $Q$ has fixed nonzero sign on $S$, it does not vanish on $S$. Therefore $f=P/Q$ is continuous, and compactness of $S$ implies that $f$ attains its maximum.

We reduce this rational optimization problem to polynomial optimization by introducing one additional variable. Define
\[
\widehat S=\bigl\{(x,y,z)\in\mathbb R^3:(x,y)\in S,\ zQ(x,y)-P(x,y)=0\bigr\}.
\]
For every $(x,y)\in S$, the equation has the unique solution $z=P(x,y)/Q(x,y)=f(x,y)$. Thus $\widehat S$ is exactly the graph of $f$ over $S$. Since $f$ is continuous, the map
$(x,y)\mapsto(x,y,f(x,y))$ is continuous, and hence $\widehat S$ is compact. We now consider the polynomial objective
\[
        F(x,y,z)=-z.
\]
By the definition of $\widehat S$, every feasible point satisfies $z=f(x,y)$, so
\[
 \min_{(x,y,z)\in\widehat S}F(x,y,z)=-\max_{(x,y)\in S}f(x,y).
\]
Consequently, if $(x^\star,y^\star,z^\star)$ minimizes $F$ over $\widehat S$, then $(x^\star,y^\star)$ maximizes $f$ over $S$, with maximum value $z^\star$.

The set $\widehat S$ is semialgebraic. The additional equation $zQ-P=0$ may be written as the two weak inequalities $zQ-P\ge0$, and $P-zQ\ge0$. Hence $\widehat S$ is described by at most $r+2$ weak polynomial inequalities. These two new inequalities involve opposite signs of the same polynomial, so the description contains at most $r+1$ distinct polynomials. The original defining polynomials have degree at most $d$, while $zQ-P$ has degree at most $d+1$.

We apply the global-optimization algorithm of Basu et al.~\cite[Algorithm~14.9 and its complexity analysis, pp.~557--558]{basu2006algorithms} to $F$ over $\widehat S$. For a
semialgebraic set described by $s$ polynomials of degree at most $\Delta$ in $k$ variables, the algorithm computes an attained infimum and a minimizing point using
$s^{\,2k+1}\Delta^{O(k)}$ arithmetic operations. In our lifted problem, $k=3$, $s\le r+1$, and $\Delta\le d+1$. The minimum is attained because $\widehat S$ is compact. Thus, the number of arithmetic operations is bounded by $(r+1)^7(d+1)^{O(1)}$. Since $r$ and $d$ are fixed constants, this is $O(1)$.

For the tie-breaking assertion, let $M=\max_S f$, and suppose the preferred subset is defined by the fixed strict inequalities $g_1,\ldots,g_t>0$. On the maximum level set we have $P-MQ=0$. We introduce one further variable $\eta$ and impose 
\[
        0\le\eta\le1,    \qquad     g_i-\eta\ge0    \quad(1\le i\le t).
\]
Together with $(x,y)\in S$ and $P-MQ=0$, these constraints define a compact semialgebraic set. We first test this set for emptiness. If it is nonempty, we maximize $\eta$ over it. The optimum is positive when the maximum level set contains a point satisfying all the strict inequalities: any preferred point admits some $\eta>0$, while $\eta>0$ implies $g_i>0$ for every $i$. Hence a maximizing point with $\eta>0$ gives the preferred witness.

The emptiness test and the subsequent optimization both involve only a constant number of variables, constraints, and bounded-degree polynomials. Thus the same fixed-dimensional semialgebraic algorithms perform both operations in $O(1)$ time.

Finally, Basu et al. count arithmetic operations over the underlying ordered domain; bit complexity for integer inputs is treated separately~\cite[Sec.~8.1, p.~285]{basu2006algorithms}. In our application, the degrees of all algebraic numbers involved are bounded solely as a function of the fixed constants $r$ and $d$. Their isolation and comparison
therefore take constant time under the algebraic real-RAM assumptions stated above. Thus both the maximum value and an attaining point, with the stated preference whenever applicable, can be computed in $O(1)$ time.~\hfill$\qed$
\end{proof}

\begin{observation}\label{obs:vertical-closure}
If $(x,y)\in\mathfrak T$ and $0\le y'\le y$, then $(x,y')\in\mathfrak T$.
\end{observation}

\begin{proof}
Because $\mathfrak T$ is $x$-monotone and its lower boundary is the base segment, its vertical section at abscissa $x$ is $\{x\}\times[0,T(x)]$. Hence every point vertically below a point of the terrain and above $\Pi$ also belongs to the terrain.~\hfill$\qed$
\end{proof}

\section[Computation of Q-star]{Computation of \texorpdfstring{$Q^*$}{Q-star}}\label{sec:quad}

We now turn to the exact algorithm. The development has two parts. We first
derive structural properties of an optimum and show that every relevant
base-incident side belongs to one of a small number of geometric types.
We then design a separate optimization procedure for each type and combine
the resulting candidates.

For a grounded quadrilateral, the side lying on the terrain base is called
its \emph{base}. We list the vertices clockwise starting with the left base
endpoint. A grounded quadrilateral
$Q=(q_1,q_2,q_3,q_4)$ is \emph{locally area-maximal} if there exists
$\varepsilon>0$ such that $\area(Q')\le\area(Q)$ for every feasible
grounded quadrilateral
$Q'=(q'_1,q'_2,q'_3,q'_4)$ satisfying
$\max_{1\le i\le4}|q_iq'_i|<\varepsilon$. Every globally optimal grounded
quadrilateral is locally area-maximal in this sense.

\begin{observation}\label{obs:boundary-contact}
The top side of every locally area-maximal grounded quadrilateral meets the upper terrain chain $\Gamma$.
\end{observation}

\begin{proof}
Let $e$ be the top side of a locally area-maximal grounded quadrilateral and suppose that $e\cap\Gamma=\varnothing$. Since $e$ and $\Gamma$ are compact, $\chi:=\operatorname{dist}(e,\Gamma)>0$. Keep the other three supporting half-planes fixed and translate the top supporting line outward. As the translation distance tends to zero, the perturbed quadrilateral converges to the original one, and its newly added portion converges to $e$. Hence, for a sufficiently small translation, the clearance $\chi$ guarantees that the perturbed
quadrilateral remains contained in $\mathfrak T$. Since $e$ has positive length, the translation strictly increases the area, contradicting local area maximality.~\hfill$\qed$
\end{proof}

The same outward-translation idea applies to a left or right side whenever
its intersection with the base can move locally. A side rooted exactly at a
terrain-base endpoint is different: the endpoint cone itself may block the
motion. Definition~\ref{def:active-extremal} incorporates this possibility.
Throughout the structural discussion, an upper-chain contact at a
quadrilateral vertex counts as a contact of both incident sides, and an
overlap of a quadrilateral side with a terrain edge is also regarded as a
\emph{side contact}.

\subsection[Properties of Q-star]{Properties of \texorpdfstring{$Q^*$}{Q-star}}\label{sec:quad_prop}

\begin{lemma}\label{lem:grounding}
The following properties hold.
\begin{enumerate}
\item[(i)] The two interior base angles of every locally area-maximal grounded quadrilateral lie in $(0,90^\circ]$.
\item[(ii)] Among the area maximizers in $\overline{\mathcal Q(\mathfrak T)}$, at least one is grounded. Moreover, if the optimum value is attained by a nondegenerate quadrilateral, then some grounded nondegenerate quadrilateral attains the same optimum value.
\end{enumerate}
\end{lemma}

\begin{proof}
We first prove part~\textup{(i)}. Let $A,B,C,D$ be a locally area-maximal grounded quadrilateral in clockwise order, with $\overline{AD}\subseteq\Pi$ and $A$ the left base endpoint. Suppose $B_x<A_x$. For $0<\delta<A_x-B_x$, let $A_\delta=(A_x-\delta,0)$. For sufficiently small $\delta$, strict convexity is preserved, and the only added region is
$\triangle A_\delta AB$. Every point of this triangle lies vertically below a point of $\overline{AB}$, so feasibility follows from Observation~\ref{obs:vertical-closure}. Moreover,
\[
 \area(A_\delta BCD)-\area(ABCD)=\frac{\delta B_y}{2}>0,
\]
contradicting local maximality. Hence $B_x\ge A_x$. Symmetrically, $C_x\le D_x$. Since $B_y,C_y>0$, these are exactly the two claimed base-angle bounds.

We now prove part~\textup{(ii)}. If no quadrilateral attains the optimum,
then the optimum is triangular, and the grounding lemma of Cabello
et al.~\cite[Lemma~1]{CabelloDDM25} gives a grounded maximum-area terrain
triangle.

Suppose instead that an optimum quadrilateral $Q$ exists. Translate $Q$
vertically downward until it first meets $\Pi$. Vertical closure preserves
containment and translation preserves area. If more than one vertex lies on
$\Pi$, then the intersection of the convex quadrilateral with the supporting
line $\Pi$ is a face, so $Q$ is grounded. We may therefore assume that a
unique lowest vertex $A$ lies on $\Pi$. Write the vertices clockwise as
$A,B,C,D$; then $B_y,D_y>0$.

For $X\in\{B,D\}$ and $t\in[0,1]$, let
\[
 X_t=(X_x,(1-t)X_y),\qquad
 Q_X(t)=\operatorname{conv}
 \bigl((\{A,B,C,D\}\setminus\{X\})\cup\{X_t\}\bigr).
\]
We claim that
\[
 Q_X(t)\subseteq\mathfrak T
 \qquad (X\in\{B,D\},\;0\le t\le1).                    \tag{G1}
\]
Indeed, if $p\in Q_X(t)$ has coefficient $\alpha$ at $X_t$ in a convex
representation, replace $X_t$ by $X$ with the same coefficients to obtain
$\widetilde p\in Q$. Then
$p_x=\widetilde p_x$ and
$\widetilde p_y-p_y=\alpha tX_y\ge0$. Since $p_y\ge0$, vertical closure
implies $p\in\mathfrak T$, proving~\textup{(G1)}.

At $t=0$ all four consecutive orientation determinants are strictly
negative. Since each is affine in $t$ after lowering either $B$ or $D$,
there exists $\varepsilon>0$ such that both $Q_B(t)$ and $Q_D(t)$ retain
their clockwise convex order for $0\le t<\varepsilon$. The shoelace formula
gives
\[
 \area(Q_B(t))-\area(Q)
   =\frac{tB_y}{2}(A_x-C_x),                           \tag{G2}
\]
and
\[
 \area(Q_D(t))-\area(Q)
   =\frac{tD_y}{2}(C_x-A_x).                           \tag{G3}
\]
Thus $C_x<A_x$ contradicts optimality via~\textup{(G2)}, while $C_x>A_x$ contradicts it via~\textup{(G3)}. Hence $C_x=A_x$, so the diagonal $\overline{AC}$ is vertical.

Equation~\textup{(G2)} shows that lowering $B$ preserves area as long as convexity is maintained. From $\operatorname{orient}(A,B,C)<0$ and $\operatorname{orient}(C,D,A)<0$, using $A_y=0$ and $C_x=A_x$, we obtain $B_x<A_x<D_x$. Set $\lambda=\frac{A_x-B_x}{D_x-B_x}\in(0,1)$, and $z_t=(1-\lambda)B_t+\lambda D$. Then $(z_t)_x=A_x$, so $z_t$ lies on the vertical line through $\overline{AC}$, and
\[
 (z_t)_y=(1-\lambda)(1-t)B_y+\lambda D_y.
\]
At $t=0$, $z_0$ is the intersection of the diagonals of $Q$, hence
$0<(z_0)_y<C_y$. Since $(z_t)_y$ decreases with $t$ and
$(z_1)_y=\lambda D_y>0$,
\[
 0<(z_1)_y\le (z_t)_y\le (z_0)_y<C_y
 \qquad (0\le t\le1).
\]
Thus $z_t$ remains in the relative interior of both
$\overline{AC}$ and $\overline{B_tD}$, so
$A,B_t,C,D$ remains a convex quadrilateral throughout the motion.

By~\textup{(G1)}, every $Q_B(t)$ is feasible, and by
$C_x=A_x$ and~\textup{(G2)} it has the same area as $Q$. At $t=1$,
$B_1=(B_x,0)$ and $A$ both lie on $\Pi$, so $Q_B(1)$ is a grounded optimum
with base $\overline{B_1A}$. Since $B_x<A_x<D_x$ and $D_y>0$, its interior
base angle at $A$ is strictly larger than $90^\circ$, contradicting
part~\textup{(i)}.

Hence an optimum quadrilateral cannot have a unique lowest vertex. Therefore
some optimum has a side on $\Pi$, and is grounded.
\qed
\end{proof}

Lemma~\ref{lem:grounding} lets us work from now on with a grounded
optimum. Part~\textup{(i)} restricts the geometry at the two base vertices,
while part~\textup{(ii)} guarantees that imposing a base side does not lose
the optimum value. We therefore write a grounded quadrilateral clockwise as
$A,B,C,D$, with $A$ the left base endpoint and
$\overline{AD}\subseteq\Pi$. We call $\overline{AB}$ the \emph{left side},
$\overline{BC}$ the \emph{top side}, and $\overline{CD}$ the
\emph{right side}.

The next step is to understand how a locally optimal left or right side can
be supported by the terrain. We keep the other three supports fixed and
rotate the actual supporting line of the side under consideration. The two
definitions below distinguish an obstruction near the root of the side from
one that occurs near its upper endpoint.

\begin{definition}\label{def:active-extremal}
Let $S$ be the left or right side of a grounded quadrilateral, with base endpoint $b$ and upper endpoint $p$, oriented from $b$ to $p$; thus $(b,p)=(A,B)$ on the left and $(b,p)=(D,C)$ on the right. If the relative interior of $S$ meets the upper terrain chain, let $v$ be the last such contact encountered from $b$ toward $p$, representing an overlap with a terrain edge by its endpoint closer to $p$. If there is no such contact, set $v=b$. We call $v$ the \emph{pivot}. 

Let $\theta_0$ be the current direction of $S$. Rotate its supporting line about $v$. If $v\ne b$, let $b(\theta)$ be the intersection of the rotated line with the base line; call $\theta$ \emph{root-feasible} if $b(\theta)\in\overline{B_\ell B_r}$ and $\overline{b(\theta)v}\subseteq\mathfrak T$. If $v=b$, keep $b$ fixed; call $\theta$ \emph{root-feasible} if a sufficiently short initial segment of the corresponding ray lies in $\mathfrak T$. 

Let $I_S$ be the maximal interval containing $\theta_0$ whose angles are root-feasible. The side $S$ is \emph{active-extremal} if $\theta_0\in\partial I_S$, including the case $I_S=\{\theta_0\}$.
\end{definition}

Active extremality therefore records an obstruction at the root of the side; contacts of the supporting line beyond the used side are irrelevant. Even when $\theta_0$ lies in the interior of $I_S$, however, a feasible rotation may fail near the upper corner, because the remaining side segment or a newly added portion of the top may leave the terrain. Definition~\ref{def:corner-blocked} describes this second obstruction.

\begin{definition}\label{def:corner-blocked}
Let $S$ have base endpoint $b$, upper endpoint $p$, pivot $v$, and root-feasible interval $I_S$, with $\theta_0\in\operatorname{int} I_S$. Keeping the base, top support, and opposite side support fixed, rotate $S$ about $v$, and let $p(\theta)$ be its intersection with the top support. The side $S$ is \emph{corner-blocked} if arbitrarily small root-feasible
rotations in at least one direction make the quadrilateral infeasible because either $\overline{vp(\theta)}$ or the newly added portion of the top side leaves $\mathfrak T$.
\end{definition}

Thus active extremality obstructs root feasibility, whereas corner blocking obstructs extension of a root-feasible rotation at the upper corner. Every corner-blocked side has its upper endpoint on the upper terrain chain. Indeed, if $p$ were interior to $\mathfrak T$, then a neighborhood of $p$ would lie inside $\mathfrak T$, so sufficiently small motions of $p$ and of the newly added top-side portion would remain feasible.

By the choice of the pivot $v$, the open segment $(v,p)$ has no upper-chain contact and hence lies in the interior of $\mathfrak T$. Any closed subsegment bounded away from $v$ and $p$ therefore has positive clearance from the terrain boundary and remains feasible under sufficiently small rotations. Near $v$, feasibility follows directly from root feasibility when
$v=b$. Otherwise, the direction from $v$ toward $p$ lies strictly inside the local terrain wedge: a boundary direction would follow a terrain edge for positive length, contradicting the choice of $v$ as the last contact. Hence every sufficiently small root-feasible rotation preserves both the moving side and the newly added top portion, contradicting corner blocking.

An upper endpoint on the upper terrain chain is called an \emph{exposed corner}. Thus every corner-blocked side has an exposed corner, although an exposed corner need not be corner-blocked. A left or right side is \emph{balanced at $v$} if $v$ lies on the upper terrain chain and is the midpoint of the side; $v$ is its \emph{balanced centre}. Additional contacts
are allowed at boundary events. The next lemma gives the resulting classification and, when a side is neither active-extremal nor corner-blocked, identifies its unique interior contact.

\begin{lemma}\label{lem:twopoint}
For a locally area-maximal grounded quadrilateral with positive-slope top, the right side $DC$ must be active-extremal or corner-blocked. If the left side $AB$ is neither, then it is balanced at a terrain vertex, its unique upper-chain contact in the relative interior of $AB$. The negative-slope case is symmetric. For a horizontal top, the same conclusion holds for
either side.
\end{lemma}

\begin{figure}[t]
\centering
\includegraphics[scale=0.7]{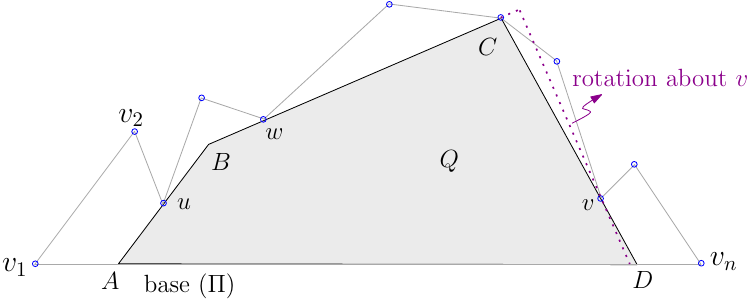}
\caption{A root-feasible right-side rotation blocked at the exposed corner $C$.}
\label{fig:tp}
\end{figure}

\begin{proof}
Let us first assume that the top edge $BC$ has positive slope (see Figure~\ref{fig:tp}). We denote its supporting line by $\tau$ and set $e=\tau\cap\Pi$. Since $\tau$ supports the quadrilateral from above, $e$ lies to the left of both base endpoints. Hence 
\begin{equation}\label{eq:twopoint-decomposition}
   \area(ABCD)=\area(\triangle eDC)-\area(\triangle eAB).
\end{equation}
Thus a rotation of the right side changes the first triangle, whereas a rotation of the left side changes the subtracted one.

\medskip
\noindent\emph{The right side $DC$.}
We denote by $v$ the pivot of $DC$. If the current direction is an endpoint of the root-feasible interval, then $DC$ is active-extremal; if an arbitrarily small root-feasible rotation fails at $C$, then it is corner-blocked. We therefore consider the remaining case. Since $\theta_0\in\operatorname{int}I_S$ and the side is not corner-blocked, both sufficiently small rotations about $v$ yield feasible quadrilaterals

We translate the $x$-coordinate so that $v=(0,s)$ and consider the rotating support $x=k(y-s)$, where $k=\cot\theta$. We write the fixed top support as $\tau:x=c+ay$, where $a>0$. Its intersections with the base and the rotating support are
\[
 D(k)=(-ks,0),\qquad C_x(k)=\frac{k(c+as)}{k-a},\qquad  C_y(k)=\frac{c+ks}{k-a},
\]
and $e=(c,0)$. We consequently obtain
\[
 \area(\triangle eD(k)C(k)) =\frac{(c+ks)^2}{2(a-k)}.
\]
Since $\triangle eAB$ remains fixed,
\[
 A_R(k)   =\frac{(c+ks)^2}{2(a-k)}-\area(\triangle eAB),  \qquad  A_R''(k)=\frac{(as+c)^2}{(a-k)^3}.
\]

We first consider the case $as+c=0$. Then $\tau$ passes through $v$, so $C(k)=v$, and the variable term is a nonconstant affine function of $k$. Hence there is no interior local maximum. If $as+c\ne0$, Lemma~\ref{lem:grounding}\textup{(i)} gives $k\le0$. Since $a>0$, we have $a-k>0$, and therefore $A_R''(k)>0$. Thus the area is strictly convex and cannot attain an interior local maximum while both nearby rotations remain feasible. Hence $DC$ is active-extremal or corner-blocked.

\medskip
\noindent\emph{The left side $AB$.}
We now denote by $v$ the pivot of $AB$ and again consider the case in which the current direction lies in the interior of its root-feasible interval and both sufficiently small rotations are feasible. We translate the $x$-coordinate so that $v=(0,s)$ and write the rotating support as $x=k(y-s)$. The top support remains $\tau:x=c+ay$. The moving vertices are
then
\[
 A(k)=(-ks,0),\qquad  B_x(k)=\frac{k(c+as)}{k-a},\qquad  B_y(k)=\frac{c+ks}{k-a}.
\]
Now $\triangle eDC$ is fixed, while $\triangle eA(k)B(k)$ is the varying subtracted triangle. Therefore
\[
 A_L(k)=\area(\triangle eDC)-\frac{(c+ks)^2}{2(a-k)}, \qquad A_L''(k)=-\frac{(as+c)^2}{(a-k)^3}.
\]

By Lemma~\ref{lem:grounding}\textup{(i)}, the left side rises vertically or to the right, so $k\ge0$. The side direction from $A$ to $B$ is $(k,1)$, whereas the top direction from $B$ to $C$ is $(a,1)$. Strict convexity at $B$ requires a clockwise turn, and hence $k-a<0$. Thus $0\le k<a$ in a neighborhood of the current configuration.

If $as+c\ne0$, then $A_L''(k)<0$, so $A_L$ is strictly concave; if $as+c=0$, it is affine. In either case $A_L$ is differentiable, and an interior local maximum must satisfy
\[
 A_L'(k)=\frac{(c+ks)(ks-c-2as)}{2(k-a)^2}=0.
\]
The factor $c+ks$ cannot vanish, since this would give $B_y(k)=0$. Consequently, we have $ks=c+2as$. Substitution into the expression for $B_y(k)$ gives $B_y(k)=2s$. The collinear points $A(k),v,B(k)$ therefore have heights $0,s,2s$, respectively. Since affine position along the side varies linearly with height, $v$ is the midpoint of $AB$.

Since $B_y=2s>0$, we have $s>0$, and hence $v$ is an upper-chain contact rather than the base endpoint. We also have $k\ne0$. Indeed, if $k=0$, the side would be vertical and $B=(0,2s)$, whereas $v=(0,s)\in\Gamma$ implies $T(0)=s$. Thus $B$ would lie above the terrain at the same abscissa, contradicting containment. Hence $k>0$, and locally we consider the side as $y=L_k(x):=s+\frac{x}{k}$.

We next establish that $v$ is the unique upper-chain contact in the relative interior of $AB$. By the definition of the pivot, no such contact lies strictly between $v$ and $B$. We therefore consider a hypothetical contact $w$ between $A$ and $v$. Since $w_x<0$, $\frac{\partial L_k(w_x)}{\partial k} =-\frac{w_x}{k^2}>0$. Thus a small increase of $k$ raises the side above the terrain at $w$. That direction is therefore not root-feasible. Since $k=\cot\theta$ is a strictly monotone local parameter for the side direction, the current
direction would then lie on the boundary of the root-feasible interval, making $AB$ active-extremal. This contradicts the case under consideration. Hence no additional contact lies between $A$ and $v$.

We finally show that $v$ is a terrain vertex. If $v$ lay in the relative interior of a terrain edge, then near $v$ both that edge and the side would be affine graphs, say $T(x)$ and $L_k(x)$. Feasibility gives $T(x)-L_k(x)\ge0$ on both sides of $x=0$, while $T(0)-L_k(0)=0$. Since $T-L_k$ is affine, it must vanish identically in a neighborhood of $0$. The side would therefore overlap the terrain edge for positive length beyond $v$ toward $B$, contradicting the convention that an overlap is represented by its endpoint closer to $B$. Thus $v$ is a terrain vertex. Together with the midpoint relation and uniqueness of the interior contact, this shows that $AB$ is balanced at $v$.

\medskip
\noindent\emph{Negative and zero top slope.}
For a negative top slope, reflection in a vertical line interchanges the left and right sides and gives the symmetric statement.

We finally consider a horizontal top, say $y=H$, and either side that is neither active-extremal nor corner-blocked. We translate the $x$-coordinate so that its pivot is $v=(0,s)$ and write its support as $x=k(y-s)$, with current parameter $k_0$. Both sufficiently small changes of $k$ are feasible. At height $y$, changing $k_0$ to $k$ moves the side horizontally
by $(k-k_0)(y-s)$. Therefore
\[
 A(k)-A(k_0)   =\pm\int_0^H (k-k_0)(y-s)\,dy    =\pm\frac{(k-k_0)H(H-2s)}2,
\]
where the sign is positive for the right side and negative for the left. If $H\ne2s$, one of the two feasible perturbations increases the area, contradicting local maximality. Hence $H=2s$, so the pivot is the midpoint of the side.

Since $H>0$, we have $s>0$, and therefore $v$ is an upper-chain contact. The side cannot be vertical, because its upper endpoint would then have height $2s$ at the same abscissa where the terrain height is only $s$. The same additional-contact argument used above shows that $v$ is the unique upper-chain contact in the relative interior of the side; for the
right side the argument follows by reflection. Likewise, the same affine-overlap argument excludes $v$ from the relative interior of a terrain edge. Hence $v$ is a terrain vertex and the side is balanced. Thus, for a horizontal top, each side is active-extremal, corner-blocked, or balanced.~\hfill$\qed$
\end{proof}

In particular, once the active-extremal and corner-blocked cases are excluded, only terrain vertices need be considered as balanced centres.

\noindent\textit{From local side types to carriers.}
Active extremality is a local condition, but an algorithm needs a finite
family of candidate supports. We obtain such a family from the
base-rooted shortest-path structure of Das et al.~\cite{DasDM21} and the
forward prolongations used by Cabello et al.~\cite{CabelloDDM25}.

Let $\overline{wu}$ be the final edge of the appropriate base-rooted
shortest path, directed toward its upper-chain endpoint $u$. Continue the
ray through $\overline{wu}$ beyond $u$ for as long as it remains feasible,
stopping at the first subsequent contact $q$ with the upper chain. If the
ray leaves $\mathfrak T$ immediately, no forward carrier is stored; if it
initially overlaps a terrain edge, $q$ is taken to be the forward endpoint
of that overlap. The resulting feasible segment $\overline{uq}$, directed
away from the base root, is the \emph{forward carrier} of
$\overline{wu}$.

\begin{result}\label{lem:carrier-families}
For an upper-chain vertex $v$, the shortest path from $B_\ell$ to $v$ is the lower convex chain of the corresponding prefix; symmetrically, the path from $B_r$ is the lower convex chain of the suffix. Hence the edge entering $v$ is its lower-hull predecessor edge. The nonempty forward prolongations of these directed edges form two $O(n)$-size carrier families $\mathcal L$ and $\mathcal R$, rooted at $B_\ell$ and $B_r$. Within each family the carriers have pairwise disjoint interiors, and both families are constructible in
$O(n)$ time. Collinear contacts are represented by their forward endpoint, and empty prolongations are discarded.
\end{result}

The shortest-path characterization follows from Das et al.~\cite{DasDM21}, Lemmas~3--5, while the prolongation, slope, noncrossing, and construction properties follow from Cabello et
al.~\cite{CabelloDDM25}, Lemmas~3--4. In particular, along a right-rooted carrier, moving away from the root increases height and does not increase $x$; the left-rooted statement is symmetric.

For a nonempty carrier $\rho$, let $u(\rho)$ be the terminal vertex of the last shortest-path edge and $q(\rho)$ its first subsequent forward contact with $\Gamma$; an initial terrain-edge overlap is represented by its forward endpoint. Thus $\rho=[u(\rho),q(\rho)]$ is directed from its \emph{initial contact} $u(\rho)$ to its \emph{terminal contact} $q(\rho)$.
Extend its supporting line rootward to the terrain base, and denote the intersection by $b(\rho)$. We retain $\rho$ only when $b(\rho)_x\le u(\rho)_x$ for a left-rooted carrier and
$b(\rho)_x\ge u(\rho)_x$ for a right-rooted one, as required by Lemma~\ref{lem:grounding}\textup{(i)}. A side with upper endpoint $p$ is \emph{carrier-representable} by a retained carrier $\rho$ when $p\in(u(\rho),q(\rho)]$ and the side is $[b(\rho),p]$; then $\rho$ is its \emph{active carrier}. The portion $[b(\rho),u(\rho)]$ lies on the supporting line of the last shortest-path edge, while $[u(\rho),p]$ lies on the stored carrier. If $p=q(\rho)$, the side is also incident to an exposed corner. 

\begin{definition}\label{def:carrier-representable}
A left or right side is \emph{carrier-representable} if it has the form above for a retained carrier from the appropriate rooted family.
\end{definition}

The next lemma connects active extremality to the carrier representation used by the algorithm.

\begin{lemma}\label{lem:active-forward}
Each active-extremal left or right side of a locally area-maximal grounded quadrilateral either uses a stored carrier, or is incident to an exposed corner, or both. In the carrier case, its supporting line contains the last edge of the relevant base-rooted shortest path. Both alternatives occur when the upper endpoint of the side is the terminal contact of the carrier.
\end{lemma}

\begin{figure}[t]
\centering
\includegraphics[scale=0.8]{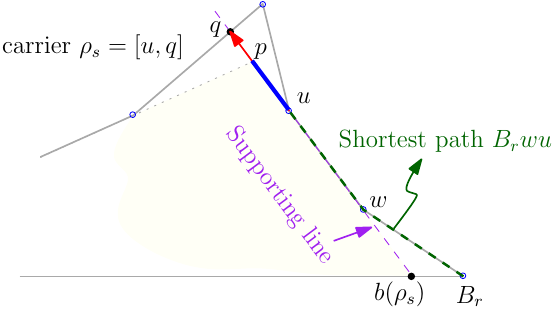}
\caption{Carrier localization. The shortest path $B_rwu$ ends with $\overline{wu}$, whose forward prolongation is the carrier $\rho=[u,q]$. The side uses $[u,p]\subseteq\rho$.}
\label{fig:active-forward}
\end{figure}

\begin{proof}
We prove the right-side statement. Let $p$ be the upper endpoint, and let $u$ be the last upper-chain contact strictly before $p$, when such a contact exists. A positive-length overlap with a terrain edge is represented by its endpoint farther toward $p$. With this convention, an isolated contact $u$ is a terrain vertex: a contact in the relative interior of a terrain edge would otherwise be transverse, contradicting feasibility on both sides of the contact. If there is no such preceding contact and $p\in\Gamma$, then $p$ is an exposed corner.

We first record a local observation about root feasibility. Suppose that the base endpoint lies in the relative interior of $\overline{B_\ell B_r}$, that there is no earlier upper-chain contact before $u$, and that the supporting line does not follow a terrain edge into $u$. Then the current direction is not active-extremal.

Indeed, near $u$ the root-side portion lies strictly inside the local feasible wedge. Away from a sufficiently small neighborhood of $u$, the segment from the base endpoint to $u$ is compact and has positive clearance from the upper chain; the base endpoint also has positive clearance from both ends of the terrain base. Hence sufficiently small rotations in both
directions preserve the entire root-side portion up to $u$. Both directions are therefore root-feasible. This observation concerns only root feasibility: a rotation may still fail near $p$, which is precisely the corner-blocked case.

Consequently, if a right side with an interior base endpoint is active-extremal, then before reaching $u$ its supporting line either has an earlier upper-chain contact $w$ or follows a terrain edge $\overline{wu}$ into $u$. Under the overlap convention, $w$ and $u$ may be taken to be terrain vertices.

We now show that $\overline{wu}$ is the lower-hull edge entering $u$ in the right-rooted suffix. Between $u$ and the base intersection, feasibility places the upper terrain chain on or above the supporting line. From the base intersection to $B_r$, the same holds because Lemma~\ref{lem:grounding}\textup{(i)} implies that a right side rises vertically or to the left. Its supporting line therefore has nonpositive slope and lies at or below the horizontal base on this interval; the terrain, being above the base, lies above it.

Thus every terrain vertex from $u$ to $B_r$ lies on or above the line through $w$ and $u$. Since $w$ and $u$ lie on this line and general position excludes a third collinear terrain vertex, $\overline{wu}$ is the lower-hull edge entering $u$. By Fact~\ref{lem:carrier-families}, it is the final edge of the shortest path from $B_r$ to $u$.

If the side is rooted at $B_r$ itself, the feasible segment $\overline{B_ru}$ is already a shortest path. If it initially follows the terrain edge incident to $B_r$, the overlap convention chooses its endpoint farther toward $p$. Hence in this case as well, the final shortest-path edge has the same supporting line as the side.

Let $[u,q]$ be the stored forward prolongation of this edge. The upper endpoint $p$ lies on this feasible continuation. If $p$ is interior to $\mathfrak T$, then it occurs before the first subsequent upper-chain contact $q$; if $p\in\Gamma$, then $p=q$. Hence $[u,p]\subseteq [u,q]$, with equality exactly when $p=q$. Thus, the side is carrier-representable; when $p=q$, it is also incident to an exposed corner.

It remains to exclude the case in which there is no upper-chain contact before $p$ and $p$ lies in the terrain interior. At an interior base point, the current direction then lies in the interior of the entering-ray cone, so it is not active-extremal. At $B_r$, an endpoint direction of that cone follows the incident upper terrain edge and therefore creates an immediate
upper-chain contact. Hence this case cannot occur. The left-side statement follows symmetrically.~\hfill$\qed$
\end{proof}

Lemma~\ref{lem:twopoint} gives the local side types, and Lemma~\ref{lem:active-forward} converts every active-extremal side into a carrier-representable side, a side incident to an exposed corner, or both. Moreover, every corner-blocked side is incident to an exposed corner by Definition~\ref{def:corner-blocked}. Thus, the possible side types are the following.

\begin{corollary}\label{cor:algorithmic-side-classification}
In a locally area-maximal grounded quadrilateral with positive top slope, the right side is carrier-representable or incident to an exposed corner; the left side is carrier-representable, balanced, or incident to an exposed corner. Negative-slope case is symmetric; for a horizontal top, either side may be balanced.
\end{corollary}

\subsection[Algorithm to compute Q-star]{Algorithm to compute \texorpdfstring{$Q^*$}{Q-star}} \label{sec:quad-algo1}

We now turn the structural classification into an algorithm. We describe the
positive-slope case first; a horizontal reflection gives the
negative-slope case. By
Corollary~\ref{cor:algorithmic-side-classification}, every base-incident
side of a relevant optimum is carrier-representable, balanced, or incident
to an exposed corner. For readability, we abbreviate
\emph{carrier-representable} to \emph{carrier} in the algorithmic
discussion.

The four resulting branches have different geometric behavior. When both
sides are carriers, the main freedom lies in the top support. A balanced
side introduces a continuous interval of possible upper endpoints, which we
reduce to auxiliary carriers. Exposed corners instead move along the upper
terrain chain and are handled by visibility and tangent sweeps. Horizontal
tops are treated last because balanced sides then admit simpler compatibility
conditions.

\paragraph{Algorithmic outline.}
We first compute the two carrier families and the maximum-area terrain
triangle $\Delta^*$. We then optimize the carrier--carrier,
balanced--carrier, exposed-corner, and horizontal-top configurations
summarized in Table~\ref{tab:positive-branches}. Each branch keeps only its
best witness. After all branches have been processed, the best
quadrilateral candidate is compared with $\Delta^*$; an equal-area tie is
resolved in favor of a nondegenerate quadrilateral. Theorem~\ref{thm:exact}
shows that these branches are exhaustive and establishes the overall
complexity.

\paragraph{Sweep conventions.}
Simultaneous events may occur even in general position. A \emph{sweep event} is a parameter value at which a terrain edge, contact, visibility state, or feasibility condition changes; between consecutive events, the defining data remain fixed. With two moving parameters, the corresponding subdivision consists of \emph{event cells}. Feasibility intervals are closed unless nondegeneracy requires a strict inequality.

Left and right sides are oriented upward from the base, the top from left to right, and carriers away from their roots. On an overlap with a terrain edge, the active contact is the endpoint farther along the oriented side. Coincident representations are coalesced. Adjacent coincident supports do not contribute a nondegenerate quadrilateral, while a collapsed top belongs to the triangular boundary represented by $\Delta^*$. At a simultaneous event, boundary candidates are evaluated before the non-strict hull
deletions and subsequent insertions. The lower hull is maintained by the Andrew monotone-chain stack~\cite{Andrew79}; each completed terrain vertex is inserted once and, after deletion, is never reinserted.

Parallel supports, visibility changes, and zeros of rational denominators
constitute boundaries of the event subdivision. Equal-area candidates are
resolved according to the quadrilateral-priority convention. If the maximum
over an event cell is attained by both degenerate and nondegenerate
candidates, the tie-breaking part of Fact~\ref{res:algebraic-cell} yields a
nondegenerate maximizer whenever one exists; we refer to this as the
\emph{maximum-level preference test}.

Table~\ref{tab:positive-branches} summarizes the positive-slope branches.
The side types are not necessarily disjoint: a balanced side may also be
carrier-representable, and a carrier may terminate at an exposed corner.
We therefore use the priority order exposed, carrier, balanced; in
particular, the balanced row represents balanced sides that are not
carrier-representable. The negative-slope branches follow by reflection,
while horizontal tops are considered separately.

\begin{table}[t]
\centering
\caption{Branch assignment for a positive-slope top.}
\label{tab:positive-branches}
\small
\renewcommand{\arraystretch}{1.12}
\begin{tabular}{c|cc}
\hline
left side $\backslash$ right side & carrier & exposed\\
\hline
carrier & carrier--carrier & carrier--exposed\\
balanced (non-carrier) & balanced--carrier & balanced--exposed\\
exposed & exposed--carrier & two exposed corners\\
\hline
\end{tabular}
\end{table}

\subsubsection{Carrier frontiers}

\paragraph{Fixed-carrier top frontier.}
Fixing the left and right carriers does not determine the quadrilateral
completely: the top support may still move while its two endpoints slide
along the carrier portions. At a local maximum this support must be blocked
by the terrain. We call a feasible top line \emph{terrain-frontier} if it
cannot be translated upward by any positive infinitesimal amount while
remaining feasible. The incidence with the finite carrier portions and the
convex support order are imposed separately.

The terminal-order lemma below shows that every non-exposed locally optimal
carrier--carrier configuration can be represented by a pair
$\ell\in\mathcal L$ and $r\in\mathcal R$ whose terminal contacts occur in
left-to-right order. Let
$p=q(\ell)\prec q=q(r)$. A top line is \emph{admissible} for this pair when
its intersections with the two supporting lines lie on the relevant finite
carrier portions and, together with the base endpoints, form a convex
quadrilateral. If $p$ lies on $\overline{v_iv_{i+1}}$ and $q$ on
$\overline{v_jv_{j+1}}$, let
$V_c=\{v_{i+1},\ldots,v_j\}$ and
$P_{\ell,r}=V_c\cup\{p,q\}$. The next lemma reduces the continuous family
of possible top supports to the lower hull of this finite point set.

\begin{lemma}\label{lem:top-frontier}
Among top lines meeting the admissible portions of $\ell$ and $r$ in convex order, the terrain-frontier supports are exactly the admissible supports of $\operatorname{LH}(P_{\ell,r})$.
\end{lemma}

\begin{proof}
Let the top support $\tau$ meet the left and right side supports at $\beta$ and $\gamma$. The carrier directions and $p\prec q$ imply 
\begin{equation}\label{eq:carrier-top-order}
        \beta_x\le p_x<q_x\le\gamma_x.
\end{equation}
We call the top portions over $[\beta_x,p_x]$ and $[q_x,\gamma_x]$ the left and right \emph{shoulders}.

We first verify that both shoulders are feasible. On $[\beta_x,p_x]$, convex support order places $\tau$ at or below the feasible left side chord. Since the vertical difference of the two
supporting lines is affine in $x$ and vanishes only at $\beta_x$, its sign cannot change before $p$; vertical closure therefore keeps the whole left shoulder inside the terrain. The right shoulder is feasible by the symmetric argument. Hence, by~\eqref{eq:carrier-top-order}, only the interval $[p_x,q_x]$ requires direct terrain testing.

On $[p_x,q_x]$, feasibility is equivalent to $\tau(x)\le T(x)$. Since both $\tau$ and $T$ are linear on each terrain edge, it suffices to test the points of $P_{\ell,r}$. Thus a feasible
terrain-frontier line must support $P_{\ell,r}$ from below: if it touched no point of $P_{\ell,r}$, it could be translated slightly upward.

Conversely, an admissible support of $\operatorname{LH}(P_{\ell,r})$ lies below the entire terrain subchain from $p$ to $q$. Together with the shoulder argument, this makes the whole top
segment feasible, while the contact with $P_{\ell,r}$ prevents any upward translation. By~\eqref{eq:carrier-top-order}, the contact lies on the actual top segment. If the contact is $p$, then $\tau(p)=T(p)=\ell(p)$. Since the top and left supports are distinct and already meet at $\beta$, we must have $\beta=p$; similarly, a contact at $q$ forces $\gamma=q$. These are exposed terminal-endpoint events. All other active contacts lie in the relative interior of the top. Therefore the terrain-frontier lines are exactly the admissible supports of $\operatorname{LH}(P_{\ell,r})$, including lower-hull edges incident to $p$ or $q$.~\hfill$\qed$
\end{proof}


\noindent\textit{Sweep implementation.} All ordering arguments use only the finite carrier portions that occur in candidates. Complete supporting lines may cross outside these portions and are never globally ordered. Within each rooted family the stored carriers have pairwise disjoint interiors~\cite{CabelloDDM25}; every later monotonicity argument is restricted to a common strip or sweep component where this local order applies.

\begin{lemma}\label{lem:terminal-order}
Let a locally area-maximal grounded quadrilateral have a nonnegative-slope top and carrier-representable left and right sides. Let $p=q(\ell)$ and $q=q(r)$ be the terminal upper-chain contacts of their active carriers. If neither top endpoint is an exposed corner, then $p\prec q$ on the upper terrain chain.
\end{lemma}

\begin{proof}
Let the top endpoints be $\beta$ and $\gamma$. Along a left-rooted carrier, the forward direction does not decrease $x$, and along a right-rooted carrier it does not increase $x$. Hence
 $\beta_x\le p_x$ and $q_x\le\gamma_x$. Let $\tau$ be the top supporting line. Convex support order puts $\tau$ strictly below the left carrier line on the open interval from $\beta_x$ to $p_x$ and strictly below the right carrier line on the open interval from $q_x$ to $\gamma_x$. The corresponding carrier chords are feasible up to their terminal contacts.

Suppose that $q\preceq p$. The two shoulder intervals $[\beta_x,p_x]$ and $[q_x,\gamma_x]$ then cover the entire projection $[\beta_x,\gamma_x]$ of the top side: $q_x\le p_x$ makes the two shoulders overlap, or meet when equality holds, so no uncovered gap remains. We now verify the perturbation at the two side intersections, where no positive-clearance argument is available. Write the top line as $y=m_Tx+c_T$, where $m_T\ge0$, and translate it upward to $\tau_\varepsilon:y=m_Tx+c_T+\varepsilon$. If the left carrier is nonvertical, write it as $y=m_Lx+c_L$. Since it lies strictly above $\tau$ immediately to the right of $\beta$, we have $m_L>m_T$, and
\[
 \frac{d}{d\varepsilon}(\tau_\varepsilon\cap\ell)_x =\frac1{m_L-m_T}>0.
\]
Thus the left intersection moves in the forward direction toward $p$.  If the right carrier is $y=m_Rx+c_R$, it lies strictly above $\tau$ immediately to the left of $\gamma$, so $m_R<m_T$, and
\[
 \frac{d}{d\varepsilon}(\tau_\varepsilon\cap r)_x =\frac1{m_R-m_T}<0;
\]
the right intersection moves in its forward direction toward $q$.  For a vertical carrier the same conclusion is immediate: the intersection keeps its $x$-coordinate and moves upward.  Because neither endpoint is terminal, for all sufficiently small $\varepsilon>0$ the new intersections $\beta_\varepsilon$ and $\gamma_\varepsilon$ remain in the relative interiors
of the two carrier segments and satisfy $(\beta_{\varepsilon})_x\le p_x$, $q_x\le(\gamma_{\varepsilon})_x$,  and $(\beta_{\varepsilon})_x<(\gamma_{\varepsilon})_x$.

The intervals $[(\beta_{\varepsilon})_x,p_x]$ and $[q_x,(\gamma_{\varepsilon})_x]$ still cover the entire new top projection, because $q_x\le p_x$.  At each coordinate in that projection, $\tau_\varepsilon$ lies on or below at least one of the two feasible carrier chords: below the left carrier to the right of $\beta_\varepsilon$, and below
the right carrier to the left of $\gamma_\varepsilon$.  Vertical closure therefore places the whole new top segment in the terrain.  Each new side consists of its unchanged feasible root-side portion up to the carrier contact and a forward subsegment of the corresponding feasible carrier; the base is unchanged.  For small $\varepsilon$, the support order remains strict, so the result is a convex quadrilateral.  Moving the top support outward strictly enlarges its area, contradicting local maximality. Therefore $p\prec q$.

If $\beta=p$ or $\gamma=q$, the corresponding endpoint lies on the upper terrain chain and the configuration belongs to an exposed-corner branch. If the two top endpoints coincide, the candidate belongs to the triangular boundary. Thus the lemma discards no case needed by the carrier--carrier branch.~\hfill$\qed$
\end{proof}


After carrier localization, every later ordering argument uses only finite,
pairwise noncrossing carrier portions; we never order complete maximal chords.
When the cross-cut lemma below is invoked, the relevant carrier portions lie
inside one of the topological-disk regions associated with a sweep component.

Let $\Omega$ be a closed region homeomorphic to a rectangle. Let $\Lambda_0$ and $\Lambda_1$ be two opposite boundary arcs of $\Omega$, and let $\Delta_0$ and $\Delta_1$ be the two remaining boundary arcs.  Orient both $\Lambda_0$ and $\Lambda_1$ from their endpoint on $\Delta_0$ toward their endpoint on $\Delta_1$. The cross-cuts considered
below join $\Lambda_0$ to $\Lambda_1$.

\begin{lemma}\label{lem:crosscut-order}
Let $\chi_1,\ldots,\chi_m$ be pairwise interior-disjoint cross-cuts of $\Omega$ joining $\Lambda_0$ to $\Lambda_1$. If their endpoints occur in the
order $\chi_1,\ldots,\chi_m$ along $\Lambda_0$, then they occur in the same order along $\Lambda_1$.
\end{lemma}

\begin{proof}
Suppose that two cross-cuts $\chi_i$ and $\chi_j$, with $i<j$, appeared in the opposite order on $\Lambda_1$. Their four endpoints would then alternate around the boundary of $\Omega$. Two arcs joining alternating boundary points of a disk must meet in its interior, contradicting the assumption that the cross-cuts have disjoint interiors. Hence the endpoint order is the same on the two boundary arcs.~\hfill$\qed$
\end{proof}

Shared boundary endpoints are treated as ties and are processed in the same sweep-event batch.

\paragraph{Changing lower hull.} 
We next explain how the lower hulls required by Lemma~\ref{lem:top-frontier} are maintained during the carrier--carrier sweep. Fix a left carrier $\ell\in\mathcal L$ with terminal contact $p=q(\ell)$, and let $r_1,\ldots,r_m$ be the relevant right carriers in the boundary order of their terminal contacts $q_1,\ldots,q_m$, where    $q_j=q(r_j)$. For the carrier pair $(\ell,r_j)$, let $P_{\ell,r_j}$ be the point set defined in the fixed-carrier frontier construction, and let $H_j=\operatorname{LH}(P_{\ell,r_j})$. By Lemma~\ref{lem:top-frontier}, the terrain-frontier top supports for $(\ell,r_j)$ are exactly the admissible supports of $H_j$. As $j$ increases, the terminal contact $q_j$ advances along the upper terrain chain, and the point set $P_{\ell,r_j}$ changes accordingly. Rather than constructing $H_j$ from scratch for every $j$, we merge the terrain vertices on the relevant upper-chain subarc with the terminal events $q_1,\ldots,q_m$ in boundary order and process this merged sequence once. An Andrew monotone-chain stack maintains the current lower hull. The following lemma gives the two properties of this sweep needed below.

\begin{lemma}\label{lem:ordered-hull-sweep}
For the sweep above, the following hold.
\begin{enumerate}
\item[(i)] When processing $r_j$, the supporting lines of the edges of $H_j$ that meet the active right carrier $r_j$ form a contiguous block in the edge order of $H_j$.

\item[(ii)] Over the complete merged scan, the Andrew stack performs $O(n+m)$ updates. In particular, every point of the merged sequence is pushed at most once and, once removed by the non-strict hull test, is never inserted again.
\end{enumerate}
\end{lemma}

\begin{proof}
For part~\textup{(i)}, fix $j$ and consider $H_j=\operatorname{LH}(P_{\ell,r_j})=(e_1,\ldots,e_t)$. Let $r_j=[w_j,q_j]$, where $q_j=q(r_j)$. Since $q_j\in P_{\ell,r_j}$, it lies on or above every supporting line of $H_j$. Such a line meets $r_j$ exactly when $w_j$ lies on or below it, with equality giving a boundary event. These are precisely the lower-hull edges visible from $w_j$, and the visible edges of a convex chain form one contiguous subchain. This proves~\textup{(i)}.

For part~\textup{(ii)}, process the terrain vertices and terminal events in their merged boundary order with the Andrew lower-hull stack. Immediately after processing $q_j$, the stack equals $H_j$. Indeed, the merged prefix contains all terrain vertices required by $P_{\ell,r_j}$; the only extra points are earlier terminal events. An earlier terminal event in the relative interior of a terrain edge that has already been completed is collinear with that edge's endpoints and cannot remain a strict lower-hull vertex. If such an event coincides with a terrain vertex, it is the same geometric point as that vertex and is coalesced in the event batch. On the terrain edge currently being scanned, the non-strict orientation rule removes earlier distinct terminal events on that edge, while coincident events are batched; immediately after processing $q_j$, the only terminal-event point from the current edge that can remain is $q_j$ itself. Hence the extra events do not change the queried lower hull.

Each item of the merged sequence is pushed at most once and, once popped, is never encountered again. Thus every item is popped at most once, and the total number of stack updates is $O(n+m)$.~\hfill$\qed$
\end{proof}




\paragraph{Carrier--carrier sweep.}
Fix $\ell\in\mathcal L$ and scan the relevant right carriers in terminal-contact order, retaining pairs with $q(\ell)\prec q(r)$ and the required base order. Lemma~\ref{lem:terminal-order} preserves every non-exposed optimum. The merged scan performs $O(n)$ hull updates, while each right carrier requires an $O(\log n)$ hull search (Lemma~\ref{lem:ordered-hull-sweep}).

\begin{lemma}\label{lem:carrier-branch}
For one fixed left carrier, all non-exposed carrier--carrier candidates can be optimized in $O(n\log n)$ time and $O(n)$ space. Over all left carriers, the total running time is $O(n^2\log n)$ and the working space  is $O(n)$.
\end{lemma}

\begin{proof}
Fix a left carrier $\ell$ and a right carrier $r$ whose carrier--carrier candidate has no exposed top corner, and let $p=q(\ell)$ and $q=q(r)$ be their terminal contacts.

\medskip\noindent\textit{Nonparallel carriers.}  Assume first that the supporting lines of $\ell$ and $r$ are not parallel. We use their intersection as the origin of an affine coordinate system, with the two carrier lines as coordinate axes directed toward the base. An admissible top line meets these axes at $(a,0)$ and $(0,b)$, where convex
support order gives $a,b>0$. Setting $X=\frac1a$, and $Y=\frac1b$, we write the top line as $Xx+Yy=1$. For a point $z=(z_1,z_2)\in P_{\ell,r}$, the top line lies below $z$ exactly 
when $Xz_1+Yz_2\le1$. Thus every point of $P_{\ell,r}$ gives a linear constraint in the $(X,Y)$-plane. For a non-exposed candidate, any contact defining an interior
piece of the feasible boundary lies in the relative interior of the actual top side, and hence has $z_1,z_2\ge0$. If $z_2>0$, its boundary is $Y=\frac{1-z_1X}{z_2}$, a nonincreasing affine function of $X$; if $z_2=0$, the boundary is vertical. Consequently, the relevant \emph{feasible top-support chain} is a nonincreasing concave polygonal chain $Y=F(X)$.

A point with a negative affine coordinate cannot define an interior edge of this chain. Indeed, if its constraint were tight for an admissible top support, the point would lie on the actual top segment. Writing $\beta$ and $\gamma$ for its endpoints, we have $\beta_x\le p_x<q_x\le\gamma_x$, while every point of $P_{\ell,r}$ lies between $p$ and $q$ in $x$-order.
Such a tight point would therefore lie between $\beta$ and $\gamma$ and have nonnegative affine coordinates, a contradiction.

\medskip\noindent \textit{Feasible part of the chain.} 
For a point $(X,F(X))$ on the chain, the top endpoints have affine coordinates $(a,0)=\left(\frac1X,0\right)$, and $(0,b)=\left(0,\frac1{F(X)}\right)$. Because $X$ increases and $F$ is nonincreasing, both endpoint positions vary monotonically. Hence the values of $X$ for which the left endpoint lies on $[u(\ell),q(\ell)]$ form an interval, as do those for which the right endpoint lies on $[u(r),q(r)]$.

The lower-hull edge supports occur in monotone slope order, so the positive-slope top supports also form a contiguous interval along the chain. Intersecting these three intervals gives a single feasible interval $\mathcal F$. Its positive-slope boundary is found by binary search on the indexed hull-edge slopes, while the points at which either top endpoint
reaches an endpoint of its finite carrier are obtained by tangent queries from that carrier endpoint to the lower hull. Thus $\mathcal F$ is located in $O(\log n)$ time.

\medskip\noindent
\textit{Maximizing the area.}
Let $o$ be the intersection of the carrier lines and let $\Theta_{\ell,r}$ be the fixed triangle bounded by those lines and the base. For top endpoints $\beta=(a,0)$ and $\gamma=(0,b)$, the quadrilateral is obtained from $\Theta_{\ell,r}$ by removing $\triangle o\beta\gamma$. Hence 
\[
        \area(\triangle o\beta\gamma)=\frac{\kappa_0ab}{2}=\frac{\kappa_0}{2XY},
\]
where $\kappa_0>0$ is the fixed area-scaling factor of the affine coordinates. Maximizing the quadrilateral area is therefore equivalent to maximizing $G(X)=XF(X)$ over $\mathcal F$.

On one edge of the support chain, write $F(X)=\alpha+\mu X$, where $\mu\le0$. Then $G(X)=\alpha X+\mu X^2$, so $G$ is concave on that edge. At a chain vertex $X_0$, concavity of $F$
gives $F'_+(X_0)\le F'_-(X_0)$, and therefore
\[
 G'_+(X_0)-G'_-(X_0)=X_0\bigl(F'_+(X_0)-F'_-(X_0)\bigr)\le0.
\]
Thus the derivative of $G$ never increases, either within an edge or across a chain vertex, and $G$ is concave on the entire feasible chain. Its values at the chain vertices are consequently unimodal, so a maximizing vertex of $\mathcal F$ is found by binary search.

We also check the interiors of the incident chain edges. On an edge with $\mu<0$, the only possible interior stationary point is $X=-\frac{\alpha}{2\mu}$. We evaluate it whenever it lies on that edge and in $\mathcal F$, together with the endpoints of $\mathcal F$. If $\mathcal F$ contains no chain vertex, it lies in one chain edge and is optimized directly. Hence a
nonparallel carrier pair is processed in $O(\log n)$ time.

\medskip\noindent\textit{Parallel carriers.} Suppose now that the carrier lines are parallel. By Lemma~\ref{lem:grounding}\textup{(i)}, a nonvertical left carrier has positive slope and a nonvertical right carrier has negative slope. Hence parallel left and right carriers must both be vertical.

Let their lines be $x=x_L$ and $x=x_R$, with $x_L<x_R$, and write the top endpoints as $\beta=(x_L,h_L)$, and $\gamma=(x_R,h_R)$. For a terrain point $z=(z_x,z_y)$ with $x_L\le z_x\le x_R$, top feasibility is equivalent to $(x_R-z_x)h_L+(z_x-x_L)h_R\le (x_R-x_L)z_y$, which is linear in $(h_L,h_R)$. The finite carrier portions give lower and upper bounds on $h_L$ and $h_R$, and the positive-slope branch additionally requires $h_L<h_R$. Thus the feasible pairs are obtained from a convex polygonal region by this slope restriction.

For such a pair the quadrilateral is a trapezoid of width $x_R-x_L$, with area $\frac{x_R-x_L}{2}(h_L+h_R)$. Hence we maximize the linear objective $h_L+h_R$. 

By Lemma~\ref{lem:top-frontier}, the relevant boundary vertices occur in lower-hull order, and the objective values are unimodal along this convex boundary chain. A binary search on the indexed chain, together with the constant number of boundary points induced by the carrier-height bounds and the slope boundary, therefore finds the optimum in $O(\log n)$ time. The
horizontal boundary is handled under the horizontal-top cases.

\medskip\noindent
\textit{All right carriers.}
We now fix $\ell$ and process the retained right carriers $r_1,\ldots,r_m$ in terminal-contact order. For $r_j$, the required lower hull is $H_j=\operatorname{LH}(P_{\ell,r_j})$. By Lemma~\ref{lem:ordered-hull-sweep}\textup{(ii)}, the merged scan of terrain vertices and terminal contacts performs $O(n)$ stack updates in total. Since there are $O(n)$ right carriers and each pair is optimized in $O(\log n)$ time, all non-exposed carrier--carrier candidates for the fixed $\ell$ are processed in $O(n\log n)$ time and $O(n)$ space.

Finally, Fact~\ref{lem:carrier-families} gives $O(n)$ left carriers. Processing them one at a time and reusing the same data structures gives $O(n^2\log n)$ total time and $O(n)$ working space. The negative-slope case follows by reflection.~\hfill$\qed$
\end{proof}


\subsubsection{Balanced left sides}\label{subsec:balanced-active}

We next handle the balanced--carrier class of
Corollary~\ref{cor:algorithmic-side-classification}. Assume that the top
edge has positive slope, the left side is balanced at the terrain vertex
$v_i=(x_i,y_i)$, and the right side is carrier-representable. Because
$v_i$ is the midpoint of the balanced side, its top endpoint must lie on
the horizontal line $H_i:\ y=h_i$, where $h_i=2y_i$. Thus a balanced centre does not determine a single side; it determines a
one-parameter family of possible top endpoints on $H_i$.

Let $p=(p_x,h_i)$ be such an endpoint. Terrain containment requires
$h_i\le T(p_x)$, so $p_x$ belongs to
$S_i:=\{x:T(x)\ge h_i\}$. The maximal closed intervals of $S_i$ are the
level components at height $h_i$. In the positive-slope case only the
components to the right of $x_i$ can contain $p$; write them as
$I_{i,1},\ldots,I_{i,t_i}$, with
$I_{i,\nu}=[a_{i,\nu},b_{i,\nu}]$. A single scan of the upper chain finds
all of them, and $t_i=O(n)$.

For later sweeps we refine each level component at any interior point or
maximal segment where $T=h_i$. The closures of the resulting intervals on
which $T>h_i$ are the sweep components, while the splitting points and
horizontal level segments are retained as closed level events. The first
task is to determine exactly which points of $H_i$ can occur as the upper
endpoint of a feasible balanced side.

For a top endpoint $p=(x_i+d,h_i)$ with $d>0$, the balanced side through $v_i$ has base endpoint $(x_i-d,0)$ and equation $y=y_i+\frac{y_i}{d}(x-x_i)$. Thus $d$ is the horizontal half-span of the balanced side. For the fixed centre $v_i$, define $d_i^{\max}=\min_{\substack{k:\ x_k<x_i\\ y_k<y_i}}\frac{y_i(x_i-x_k)}{y_i-y_k}$, where the left base endpoint is included in the minimum. Geometrically, $d_i^{\max}$ is the largest horizontal half-span allowed by feasibility of the portion of the side below $v_i$.

For a level component $I_{i,\nu}=[a_{i,\nu},b_{i,\nu}]$, define $\sigma_{i,\nu}   =\min_{x\in(x_i,a_{i,\nu}]}\frac{T(x)-y_i}{x-x_i}$. Thus $\sigma_{i,\nu}$ is the largest positive slope that the upper portion of a balanced side can have while remaining below the terrain on its way from $v_i$ to $I_{i,\nu}$.

\begin{lemma}\label{lem:balanced-interval-active}
Let $v_i=(x_i,y_i)$ be the fixed balanced centre and let $I_{i,\nu}=[a_{i,\nu},b_{i,\nu}]$ be a level component to the right of $x_i$. If $\sigma_{i,\nu}>0$, then the feasible top endpoints $p=(x_i+d,h_i)$, with $h_i=2y_i$, are exactly those for which
\[
 d\in \left[ \max\left\{ a_{i,\nu}-x_i,\, \frac{y_i}{\sigma_{i,\nu}} \right\}, \min\left\{ b_{i,\nu}-x_i,\, d_i^{\max} \right\} \right].           \tag{B5}
\]
If $\sigma_{i,\nu}\le0$, no such endpoint is feasible. For a fixed centre $v_i$, we can compute all feasible intervals in $O(n)$ time.
\end{lemma}

\begin{figure}[t]
\centering
\includegraphics[scale=0.7]{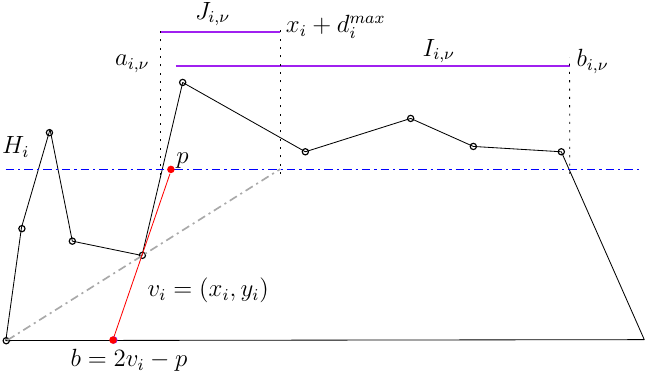}
\caption{Exact balanced-side intervals. The terrain contact $v_i$ is the midpoint of $\overline{bp}$. Here at height $h_i$, the containing level component is $I_{i,\nu}$; its feasible endpoint interval is $J_{i,\nu}$.}
\label{fig:balanced-interval}
\end{figure}

\begin{proof}
For a top endpoint $p=(x_i+d,h_i)$, with $d>0$, the balanced side has base endpoint $(x_i-d,0)$ and supporting line $L_d(x)=y_i+\frac{y_i}{d}(x-x_i)$.

We first consider the portion below $v_i$. For a terrain vertex $v_k=(x_k,y_k)$ with $x_k<x_i$ and $y_k<y_i$, feasibility requires $L_d(x_k)\le y_k$, that is, $y_i-\frac{y_i}{d}(x_i-x_k)\le y_k$. Equivalently, $d\le  \frac{y_i(x_i-x_k)}{y_i-y_k}$. Vertices with $y_k\ge y_i$ impose no stronger restriction. Since both $L_d$ and $T$ are linear on each terrain edge, checking the relevant terrain vertices is sufficient. Taking the minimum of these upper bounds, including the one induced by the left base endpoint, gives $d\le d_i^{\max}$. Thus $d_i^{\max}$ is precisely the largest horizontal half-span allowed by feasibility below $v_i$.

Now consider the portion above $v_i$. Since $T(x_i)=y_i<h_i$, every level component under consideration begins strictly to the right of $x_i$; hence $a_{i,\nu}>x_i$. Before the side
reaches $I_{i,\nu}$, feasibility is equivalent to $L_d(x)\le T(x)$, for  $x_i<x\le a_{i,\nu}$. Substituting the equation of $L_d$ gives
\[
        \frac{y_i}{d}  \le  \frac{T(x)-y_i}{x-x_i}  \qquad    \text{for all }x\in(x_i,a_{i,\nu}].
\]
By the definition of $\sigma_{i,\nu}$, these inequalities are equivalent to $\frac{y_i}{d}\le\sigma_{i,\nu}$. Hence, if $\sigma_{i,\nu}\le0$, no positive-slope balanced side can reach
the component. If $\sigma_{i,\nu}>0$, the condition becomes $d\ge\frac{y_i}{\sigma_{i,\nu}}$. Thus $\sigma_{i,\nu}$ is the largest positive slope permitted for the upper portion of the side before it reaches $I_{i,\nu}$.

Once the side reaches $a_{i,\nu}$, no further terrain constraint arises before its endpoint $p$. Indeed, throughout $I_{i,\nu}$ we have $T(x)\ge h_i$, whereas $L_d(x)\le h_i$ for
$x\le p_x$. Finally,
\[
        p_x=x_i+d\in I_{i,\nu} \Longleftrightarrow a_{i,\nu}-x_i \le d\le b_{i,\nu}-x_i.
\]
Intersecting this component constraint with $d\le d_i^{\max}$ and $d\ge\frac{y_i}{\sigma_{i,\nu}}$ gives~\textup{(B5)}.

It remains to justify the running time. A single left-to-right scan of the upper chain finds all level components, sweep components, and closed level events at height $h_i$. During the same scan, $d_i^{\max}$ is obtained by maintaining the minimum of $\frac{y_i(x_i-x_k)}{y_i-y_k}$ over the relevant vertices to the left of $v_i$. For $\sigma_{i,\nu}$, on each terrain edge the function
\[
        x\longmapsto
        \frac{T(x)-y_i}{x-x_i}
\]
is linear-fractional and has derivative of fixed sign. Its minimum on the portion of that edge under consideration is therefore attained at an endpoint of that portion, possibly at its intersection with $H_i$. Maintaining the running minimum while scanning the upper chain therefore produces every $\sigma_{i,\nu}$ in $O(n)$ total time.~\hfill$\qed$
\end{proof}


When $\sigma_{i,\nu}>0$, let $J_{i,\nu}$ be the interval of feasible top-endpoint abscissae $p_x=x_i+d$ given by Lemma~\ref{lem:balanced-interval-active}; when $\sigma_{i,\nu}\le0$, set $J_{i,\nu}=\varnothing$. In the former case, $J_{i,\nu}$ is empty precisely when
\[
\max\left\{a_{i,\nu}-x_i,\,\frac{y_i}{\sigma_{i,\nu}}\right\}>\min\left\{b_{i,\nu}-x_i,\,d_i^{\max}\right\}.
\]
The feasible top endpoints associated with $I_{i,\nu}$ are therefore the points $(x,h_i)$ with $x\in J_{i,\nu}$. An endpoint of $J_{i,\nu}$ may correspond to an additional terrain contact; such a configuration still satisfies the definition of a balanced side.


We next refine the construction to the sweep components of $I_{i,\nu}$. Let $I=[a,b]\subseteq I_{i,\nu}$ be one such component, and put $J_I:=J_{i,\nu}\cap[a,b]$. Thus $J_I$ contains exactly the abscissae of feasible balanced top endpoints lying in $I$. Consider the portion of the terrain above $H_i$ over $I$: $K_I  =\{(x,y):x\in[a,b],\ h_i\le y\le T(x)\}$. Since $T(x)>h_i$ for $x\in(a,b)$, $K_I$ is itself a $1.5$D terrain, with horizontal base $[(a,h_i),(b,h_i)]$. Any feasible positive-slope top side whose left endpoint lies in the interior of this sweep component remains in $K_I$ until its right endpoint. We apply the forward-carrier construction to $K_I$, rooted at its left base endpoint $a^*:=(a,h_i)$. The resulting forward carriers are called \emph{auxiliary top carriers}.

Let $c$ be such a carrier. Denote its supporting line by $\tau_c$, its
initial terrain contact by $w_c$, and its terminal terrain contact by
$q_c$. Extend $\tau_c$ rootward until it meets $H_i$, and denote this
intersection by $p_c$. We retain $c$ only if
\[
        (p_c)_x\in J_I\cap(a,b)        \qquad\text{and}\qquad        2\bigl((p_c)_x-x_i\bigr)m_c<h_i,                 \tag{B5a}
\]
where $m_c>0$ is the slope of $\tau_c$. The first condition makes $p_c$ a feasible balanced top endpoint in the interior of the sweep component. Let $b_c=(2x_i-(p_c)_x,0)$ be the corresponding base endpoint. Since $\overrightarrow{b_cp_c}=(2((p_c)_x-x_i),h_i)$, the second condition in \textup{(B5a)} is precisely the strict convexity condition at $p_c$.

Let $\mathcal A_i$ be the set of all retained auxiliary top carriers for the fixed centre $v_i$. The sweep-component chains have $O(n)$ total complexity, so Fact~\ref{lem:carrier-families}, after translating $H_i$ to the $x$-axis, constructs $\mathcal A_i$ in $O(n)$ time and space. Intermediate collinear contacts are removed, overlaps are represented by their forward endpoint, and identical or zero-length carriers are discarded. Thus the carriers in $\mathcal A_i$ are pairwise interior-disjoint finite segments with endpoints
on the original upper terrain chain.

\begin{lemma}\label{lem:balanced-top-carriers}
Let $Q$ be a grounded optimal quadrilateral with positive top slope, a left side balanced at $v_i$, a carrier-representable right side, and no exposed corner. If the left side of $Q$ is not carrier-representable, then the top-right vertex of $Q$ is the intersection of a carrier in $\mathcal A_i$ with the stored right carrier used by $Q$.
\end{lemma}

\begin{proof}
Let $p$ and $\gamma$ be the top-left and top-right vertices of $Q$, and let $r$ be the stored right carrier containing $\gamma$. Since the left side is balanced at $v_i$, we have $p_y=h_i$. The top has positive slope, so every point of its relative interior to the right of $p$ has height greater than $h_i$. Feasibility therefore gives $T(x)>h_i$ for $(p_x<x<\gamma_x)$. Since neither $p$ nor $\gamma$ is exposed, $T(p_x)>h_i$ and $T(\gamma_x)>\gamma_y>h_i$. Hence, $T(x)>h_i$ throughout $[p_x,\gamma_x]$. This interval lies in a single
connected component of $\{x:T(x)>h_i\}$; let $I=[a,b]\subseteq I_{i,\nu}$ be its closure. Then $p_x,\gamma_x\in(a,b)$.

We next show that $p_x\in\operatorname{int}J_{i,\nu}$. If $p_x$ were an endpoint induced by a level-component boundary, then $p$ would be exposed. If it were induced by $d_i^{\max}$, then either the balanced-side base endpoint would be $B_\ell$ or the portion below $v_i$ would acquire an additional terrain contact. Likewise, equality $d=y_i/\sigma_{i,\nu}$ creates a contact or terrain-edge overlap above $v_i$. In each of the latter cases the left side is active-extremal, and Lemma~\ref{lem:active-forward} makes it carrier-representable or incident
to an exposed corner, contrary to the assumptions. Thus $p_x\in\operatorname{int}J_{i,\nu}$. Since also $p_x\in(a,b)$, $p_x\in\operatorname{int}J_I$, and $J_I=J_{i,\nu}\cap[a,b]$.

Because the endpoints of $r$ lie on $\Gamma$ and $\gamma$ is not exposed, $\gamma$ lies in the relative interior of $r$. Moreover, all four consecutive orientation inequalities of $Q$ are strict. Hence the carrier-incidence condition at $\gamma$ and strict convexity persist under sufficiently small perturbations.

By Observation~\ref{obs:boundary-contact}, the top side meets $\Gamma$ in its relative interior. Suppose first that it has exactly one such contact, say $\omega$. Then $\omega$ is a terrain vertex. Indeed, if $\omega$ lay in the relative interior of a terrain edge, the difference between the terrain height and the top support would be affine and nonnegative on that edge, with a zero at an interior point. It would therefore vanish on a positive-length subsegment, contradicting uniqueness of $\omega$.

Let $m_-$ and $m_+$ be the slopes of the terrain edges incident to $\omega$, and let $m$ be the top-edge slope. Local feasibility gives $m_-\le m\le m_+$. Equality on either side would again create a positive-length overlap, so $m_-<m<m_+$. Thus all sufficiently close top-support slopes remain feasible in a fixed neighborhood of $\omega$.

Let $L_0=\operatorname{line}(p\gamma)$ and rotate it about $\omega$. For a nearby orientation $\theta$, let $p(\theta)=L_\theta\cap H_i$, and $\gamma(\theta)=L_\theta\cap\operatorname{line}(r)$, and set $A(\theta)=\bigl(2x_i-p(\theta)_x,0\bigr)$. Then $Q(\theta)=\bigl(A(\theta),p(\theta),\gamma(\theta),b_r\bigr)$ keeps the left side balanced at $v_i$ and the right side on the fixed carrier support. Since $p_x\in\operatorname{int}J_I$, $\gamma\in\operatorname{relint}r$, and the convexity inequalities are
strict, these conditions persist for all sufficiently small rotations in both directions.

Containment also persists. Near $\omega$ this follows from $m_-<m<m_+$; away from small neighborhoods of $\omega,p,\gamma$, the original top side has positive clearance from $\Gamma$; and near $p$ and $\gamma$ it follows because neither endpoint is exposed. Hence $Q(\theta)$ is feasible for every $\theta$ in some open interval about the
original orientation. Thus $Q$ lies in the interior of a two-sided feasible top-support perturbation.

Let $t_r=\operatorname{line}(r)\cap H_i$, $b_r=b(r)$. For every quadrilateral in this family, the portion below $H_i$ is a trapezoid whose lower and upper widths are
$(b_r)_x-A(\theta)_x$ and $(t_r)_x-p(\theta)_x$. Since $A(\theta)_x=2x_i-p(\theta)_x$, its area is independent of $\theta$:
\[
        W_r=\frac{h_i}{2}\bigl((b_r)_x+(t_r)_x-2x_i\bigr).                \tag{B5b}
\]

We now use affine coordinates with origin $t_r$, positive $\xi$-axis along $H_i$ toward the left, and positive $\eta$-axis along $\operatorname{line}(r)$ upward. Write a nearby top support as $X\xi+Y\eta=1$. Its axis intercepts are $(1/X,0)$ and $(0,1/Y)$, so the triangular portion above $H_i$ has area $\kappa/(2XY)$ for a fixed scale factor $\kappa>0$.

Let $(\xi_\omega,\eta_\omega)$ be the coordinates of $\omega$. Since $\omega$ lies in the relative interior of the top side, $\xi_\omega,\eta_\omega>0$. Every rotated support through $\omega$ satisfies $X\xi_\omega+Y\eta_\omega=1$, or equivalently 
\[Y=\alpha+\mu X, \quad \alpha=\frac1{\eta_\omega}>0,       \quad    \mu=-\frac{\xi_\omega}{\eta_\omega}<0.
\]
Therefore 
\[
        \area(Q(X))=W_r+\frac{\kappa}{2X(\alpha+\mu X)}.            \tag{B5c}
\]
The original quadrilateral corresponds to an interior feasible value $X_0$. On the feasible interval, $X(\alpha+\mu X)>0$ is a strictly concave quadratic. Its reciprocal has no
interior local maximum: its only stationary point is a strict minimum. Hence $Q(X_0)$ cannot be locally area-maximal, contradicting optimality.  The top edge must therefore have at least two contacts with $\Gamma$ in its relative interior, counting the endpoints of a terrain-edge overlap as contacts. Since neither top vertex is exposed, these contacts are terrain
vertices. Let $z$ and $w$ be the last two in increasing $x$-order.

The top support lies below the upper chain from $p$ to $w$. Its backward extension from $p$ toward $a^*=(a,h_i)$ lies at or below $H_i$, because the top has positive slope. Hence the line through $z$ and $w$ is a lower support of the prefix of $K_I$ ending at $w$. After removing intermediate collinear hull vertices and applying the overlap convention,
$\overline{zw}$ is the final edge of the left-rooted shortest path from $a^*$ to $w$ in $K_I$.

By the choice of $w$, the segment from $w$ to $\gamma$ is feasible and has no further upper-chain contact. Thus the forward prolongation of $\overline{zw}$ contains $\gamma$, so $\gamma$ lies on the corresponding auxiliary carrier $c$. Since the supporting line of $c$ is the top support of $Q$, its rootward extension meets $H_i$ at $p$, and hence $p_c=p$.
We already have $(p_c)_x\in J_I\cap(a,b)$, and strict convexity at $p$ gives $2\bigl((p_c)_x-x_i\bigr)m_c<h_i$. Therefore $c$ satisfies~\textup{(B5a)} and belongs to $\mathcal A_i$.
Since $\gamma$ lies on both $c$ and the stored right carrier $r$, it is their intersection.~\hfill$\qed$
\end{proof}

Consider a relative-interior intersection $\gamma=c\cap r$ of an auxiliary carrier $c\in\mathcal A_i$ and a right carrier $r$. Let $t_r=\operatorname{line}(r)\cap H_i$ and $b_r=b(r)$, and let $b_c$ be the balanced-side base endpoint defined above. We associate with $(c,r)$ the quadrilateral $Q_r(c)=(b_c,p_c,\gamma,b_r)$. By construction, $\overline{b_cp_c}$ is a feasible balanced side, $\overline{p_c\gamma}$ lies on the feasible auxiliary carrier $c$, and $\overline{\gamma b_r}$ lies on the feasible right carrier $r$. The second condition in~\textup{(B5a)}, together with the carrier directions and the relative-interior intersection at $\gamma$, ensures strict convexity. Hence $Q_r(c)$ is a feasible nondegenerate convex quadrilateral. Let $W_r=\frac{h_i}{2}\bigl((b_r)_x+(t_r)_x-2x_i\bigr)$ be the area of the portion of $Q_r(c)$ below $H_i$. Then
\[
        F_r(c)=\area(Q_r(c))=W_r+\area(\triangle p_ct_r\gamma).             \tag{B5d}
\]
Here $t_r$ need not lie on the finite carrier $r$; it is used only for this area decomposition. By Lemma~\ref{lem:balanced-top-carriers}, every non-exposed balanced--carrier optimum whose balanced side is not carrier-representable appears as some $Q_r(c)$. If the balanced side is also carrier-representable, the carrier--carrier branch covers it. Thus it remains to maximize $F_r(c)$ over all valid auxiliary-carrier/right-carrier intersections.

\begin{lemma}\label{lem:balanced-support-difference}
Fix a balanced centre $v_i$ with level $H_i$, and let $S=[L,R]\times\mathbb R$. Let distinct carriers $c,d\in\mathcal A_i$ meet $S$, with rootward intersections $p_c,p_d\in H_i$ and initial contacts $w_c,w_d$ satisfying $(w_c)_x,(w_d)_x\le L$. Since the carriers are interior-disjoint, their vertical order in $S$ is fixed; assume $c$ lies below $d$ at $x=L$. Let
$\tau_c=\operatorname{line}(c)$, $\tau_d=\operatorname{line}(d)$, and $z=\tau_c\cap\tau_d$. If a right carrier $r$ meets both $c$ and $d$ in their relative interiors inside $S$, with $c_r=c\cap r$ and $d_r=d\cap r$, then $z_x\le L$, $z_y\ge h_i$, $(c_r)_x$, $(d_r)_x>z_x$, and
\[
F_r(d)-F_r(c) =\area(\triangle zc_rd_r)-\area(\triangle zp_cp_d).       \tag{B5e}
\]

If another right carrier $r'$ meets both $c$ and $d$ inside $S$, define $c_{r'}=c\cap r'$ and $d_{r'}=d\cap r'$. If $r$ and $r'$ both reach $x=R$ and $r$ lies below $r'$ on their common projection, then $|zc_r|\le |zc_{r'}|$, and $|zd_r|\le |zd_{r'}|$.
\end{lemma}

\begin{figure}[t]
\centering
\includegraphics[scale=0.7]{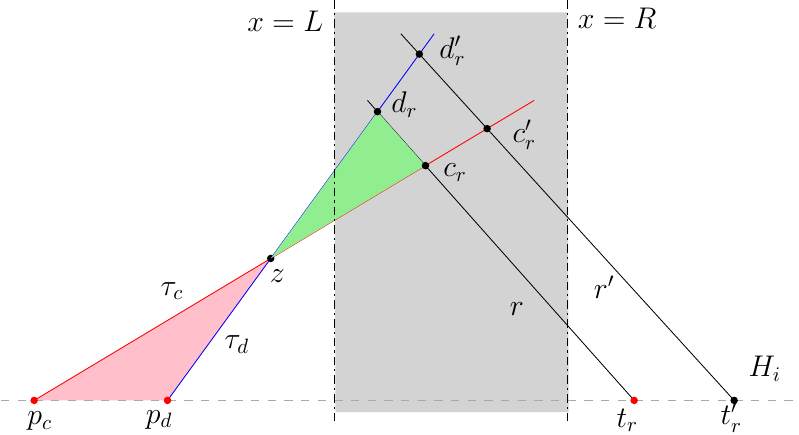}
\caption{The support-difference identity in a common strip $S=[L,R]\times\mathbb R$. The point $z$ lies to the left of all four carrier intersections. Replacing $r$ by the higher carrier $r'$ moves both intersections to larger $x$-coordinates along the fixed supports, while $\triangle z p_c p_d$ remains unchanged.}
\label{fig:balanced-support-difference}
\end{figure}

\begin{proof}
Since $(w_c)_x,(w_d)_x\le L$ and each of $c$ and $d$ meets $r$ in the interior of $S$, both finite carriers extend from the left boundary into $S$. Their $x$-projections therefore overlap on a positive-length interval to the right of $L$. Auxiliary carriers belonging to distinct height-$h_i$ sweep components have disjoint interior projections, so $c$ and $d$ belong to the same sweep component (see Figure~\ref{fig:balanced-support-difference}).

We first compare their initial contacts. We claim that $(w_c)_x\le (w_d)_x$. Suppose instead that $(w_d)_x<(w_c)_x$. Since $d$ meets the strip and later meets $r$ in its relative interior, its terminal contact $q_d$ lies strictly to the right of $L$. Thus
\[
        (w_d)_x<(w_c)_x\le L<(q_d)_x,
\]
so the point of $d$ with abscissa $(w_c)_x$ lies in the relative interior of the finite carrier $d$.

By definition, $q_d$ is the first upper-chain contact encountered after $w_d$ along the auxiliary carrier. Hence this interior point of $d$ lies strictly below the upper terrain chain. Since $w_c$ itself is an upper-chain contact of $c$, $\tau_d((w_c)_x) <T((w_c)_x)=\tau_c((w_c)_x)$. On the other hand, by the chosen order at the left boundary, $\tau_c(L)\le\tau_d(L)$. If $(w_c)_x=L$, these two inequalities already contradict each other. Otherwise continuity forces $\tau_c$ and $\tau_d$ to intersect at an abscissa in $((w_c)_x,L]$. At that abscissa both finite carriers are present, and the intersection lies in their relative interiors, contradicting the pairwise interior-disjointness of the members of $\mathcal A_i$. Therefore $(w_c)_x\le(w_d)_x$.

We now evaluate $\tau_d$ at $(w_c)_x$. There are only two cases. (i) If $(w_c)_x<(p_d)_x$, then $\tau_d$ has positive slope and $\tau_d(p_d)=h_i$, so  $\tau_d((w_c)_x)<h_i\le T((w_c)_x)$. (ii) Otherwise, $(p_d)_x\le(w_c)_x\le(w_d)_x$. The point of $\tau_d$ at abscissa $(w_c)_x$ then lies on the feasible rootward segment $\overline{p_dw_d}$, and hence $\tau_d((w_c)_x)\le T((w_c)_x)$. 
Thus in either case
\[
        \tau_d((w_c)_x)\le T((w_c)_x)=\tau_c((w_c)_x).                                  \tag{*}
\]

Since $c$ lies below $d$ at $x=L$ (with the order immediately to the right of $L$ used to break a possible equality), the affine function $\tau_d-\tau_c$ is positive immediately to the right of $L$. Together  with~\textup{(*)}, this shows that the two supporting lines meet at a unique point $z=\tau_c\cap\tau_d$ with $(w_c)_x\le z_x\le L$. Since $z_x\ge(w_c)_x$, $z\in\tau_c$, the line $\tau_c$ has positive slope, and $(w_c)_y\ge h_i$, we obtain $z_y\ge h_i$. Moreover, $c_r$ and $d_r$ occur in the strip, to the right of $z$, so $(c_r)_x>z_x$, and $(d_r)_x>z_x$.

We next compare the intersections of $c$ and $d$ with $r$ and $r'$. Consider $c$ and write $x=(c_r)_x$, $x'=(c_{r'})_x$. Suppose for a contradiction that $x'<x$. Let $a_r$ be the leftmost abscissa of the finite carrier $r$. Since $\tau_c$ has positive slope and the supporting line of a right carrier has negative slope, $\tau_c-r$ is strictly increasing and vanishes at $x$.

If $a_r\le x'$, then $x'$ lies on the finite carrier $r$, because $r$ reaches the right boundary of $S$. Since $x'<x$, $r(x')>\tau_c(x')=r'(x')$, contradicting the assumption that $r$ lies below $r'$ on their common projection. Suppose instead that $a_r>x'$. Since $c_r\in r$, we have $x'<a_r<x$. The carrier $c$ reaches the left boundary of $S$ and contains $c_r$,
while $r'$ contains $c_{r'}$ and reaches the right boundary. Hence both $c$ and $r'$ are defined at $a_r$. Since the left endpoint of $r$ lies on the upper terrain chain, $r(a_r)=T(a_r)\ge\tau_c(a_r)$. Also, $\tau_c-r'$ is strictly increasing and vanishes at $x'$, so $\tau_c(a_r)>r'(a_r)$. Consequently $r(a_r)>r'(a_r)$, again contradicting the order of $r$ and $r'$. Therefore $(c_r)_x\le(c_{r'})_x$. The same argument applied to $d$ gives $(d_r)_x\le(d_{r'})_x$.

Since all four intersection points lie to the right of $z$ on the fixed supporting lines $\tau_c$ and $\tau_d$, distance from $z$ increases with abscissa along each line. Hence $|zc_r|\le |zc_{r'}|$, and $|zd_r|\le |zd_{r'}|$.

It remains to prove~\textup{(B5e)}. For a fixed right carrier $r$, the term $W_r$ in~\textup{(B5d)} is independent of the auxiliary carrier. Thus
\[
        F_r(d)-F_r(c)=\area(\triangle p_dt_rd_r)-\area(\triangle p_ct_rc_r).
\]
By the established order of the points, 
\[
        p_c,z,c_r\in\tau_c,\qquad   p_d,z,d_r\in\tau_d,\qquad     p_c,p_d,t_r\in H_i,
\]
while $c_r,d_r,t_r$ lie on the supporting line of $r$ (see Figure~\ref{fig:balanced-support-difference}). The two triangles have the common quadrilateral $p_d t_r c_r z$. More precisely,
$\area(\triangle p_ct_rc_r)=\area(p_d t_r c_r z)+\area(\triangle zp_cp_d)$, and $\area(\triangle p_dt_rd_r)=\area(p_d t_r c_r z)+\area(\triangle zc_rd_r)$. Subtracting these two equalities gives
\[
        F_r(d)-F_r(c)=\area(\triangle zc_rd_r)-\area(\triangle zp_cp_d),
\]
and the claim follows.~\hfill$\qed$
\end{proof}

Lemma~\ref{lem:balanced-support-difference} converts the vertical order of the carrier families into an anti-Monge inequality for the values $F_r(c)$. We now use this property to search all valid finite-carrier intersections without enumerating them individually.

\begin{lemma}\label{lem:balanced-matrix-search}
For a fixed balanced centre $v_i$, the maximum of $F_r(c)$ over all relative-interior intersections between carriers $c\in\mathcal A_i$ and the original right carriers can be found in $O(n\log n)$ time and $O(n)$ working space.
\end{lemma}

\begin{proof}
We use the hereditary segment-tree decomposition for finite-segment intersections of Cabello et al.~\cite[Section~4]{CabelloDDM25}. The decomposition ensures that every relevant intersection appears in a matrix whose rows and columns admit the order required by Lemma~\ref{lem:balanced-support-difference}.

\emph{Strip decomposition and coverage.}
Sort the abscissae of all carrier endpoints and build a balanced segment tree on the resulting elementary intervals. For a finite carrier, its canonical cover consists of the maximal tree intervals contained in its $x$-projection.

For a node $\zeta$ whose interval defines the strip $S_\zeta=[L,R]\times\mathbb R$, let $\mathcal C_\zeta$ and $\mathcal R_\zeta$ be the standard lists of auxiliary and right carriers, respectively. Let $\mathcal C_\zeta^h$ and $\mathcal R_\zeta^h$ be the corresponding hereditary lists; these contain carriers assigned to proper descendants that meet $S_\zeta$ but do not span it. At $\zeta$ we process the three pairs
\[
 (\mathcal C_\zeta,\mathcal R_\zeta),\qquad
 (\mathcal C_\zeta,\mathcal R_\zeta^h),\qquad
 (\mathcal C_\zeta^h,\mathcal R_\zeta).                 \tag{B5f}
\]

To see that these pairs cover every relative-interior intersection, let $x$ be the abscissa of such an intersection. Choose an elementary interval containing $x$, or an adjacent interval when $x$ is a split value. The canonical-cover nodes of the two intersecting carriers that contain this elementary interval are ancestors of the same leaf and are therefore comparable. If they coincide, the intersection appears in the standard--standard pair. Otherwise, at the higher node one carrier occurs in a standard list and the other in the corresponding hereditary list. These are exactly the remaining two pairs in~\textup{(B5f)}.

Strip boundaries are closed, so an intersection lying in the relative interiors of both carriers is retained even when its abscissa is a split value. Duplicate evaluations are harmless. An intersection at an actual carrier endpoint is an exposed-corner event and is handled separately.

\emph{Carrier orders and boundary incidence.}
Within each carrier family, the stored finite segments are pairwise interior-disjoint. As in Cabello et al.~\cite[Section~4.3]{CabelloDDM25}, their above/below relations can therefore be extended to a total order in $O(n\log n)$ time and $O(n)$ space. Restrict these orders to the lists at every node. We use right carriers as rows and auxiliary carriers as columns, both ordered from bottom to top. A matrix entry $(r,c)$ is valid exactly when the relative interiors of $r$ and $c$ intersect in $S_\zeta$; its value is $F_r(c)$ from~\textup{(B5d)}.

For hereditary lists we also need to know which finite carriers actually reach the appropriate strip boundary. Suppose $c\in\mathcal C_\zeta^h$ meets a standard right carrier $r\in\mathcal R_\zeta$ at abscissa $x$. If $(w_c)_x>L$, then $(w_c)_x<x$. Since $\tau_c-r$ is strictly increasing and is zero at $x$, this would give $\tau_c((w_c)_x)<r((w_c)_x)$. But $w_c$ lies on the upper terrain chain, while the standard carrier $r$ is feasible there, so $\tau_c((w_c)_x)=T((w_c)_x)\ge r((w_c)_x)$, a contradiction. Hence $(w_c)_x\le L$.
Thus the actual finite carrier $c$ reaches the left boundary of the strip. Since $c$ is hereditary and therefore does not span the strip, its terminal contact $q_c$ lies inside $S_\zeta$.

The symmetric statement holds for a hereditary right carrier. If $r\in\mathcal R_\zeta^h$ meets a standard auxiliary carrier $c$ and $u(r)$ is its initial contact, then $u(r)_x\ge R$. Otherwise, since $r-\tau_c$ is strictly decreasing after their intersection, we would obtain $r(u(r)_x)<\tau_c(u(r)_x)$, whereas $r(u(r)_x)=T(u(r)_x)\ge\tau_c(u(r)_x)$. Thus the actual finite carrier $r$ reaches the right boundary, and, because it is hereditary, its terminal contact $q(r)$ lies inside the strip. No supporting line is used to extend a stored carrier beyond its finite segment.

\emph{Valid column intervals.}
Consider first the standard--standard pair. Both carriers span $S_\zeta$. Since an auxiliary support has positive slope and a right support has negative slope, their difference is strictly increasing. Hence an auxiliary carrier $c$ intersects a row $r$ in the strip exactly when $c(L)\le r(L)$ and $c(R)\ge r(R)$. Because both lists are ordered bottom to top, the valid columns in each row form a contiguous interval $[\ell_j,u_j]$, and both $\ell_j$ and $u_j$ are nondecreasing with the row index $j$.

For $(\mathcal C_\zeta,\mathcal R_\zeta^h)$, each retained hereditary row reaches the right boundary and has its terminal terrain contact inside the strip. At that terrain contact it lies on or above every spanning auxiliary carrier. Hence intersection validity is determined solely at the right boundary: $c(R)\ge r(R)$. The valid columns therefore form a suffix, whose first index is nondecreasing with the row index.

For $(\mathcal C_\zeta^h,\mathcal R_\zeta)$, each retained hereditary auxiliary carrier reaches the left boundary and terminates at a terrain contact inside the strip. Hence validity is determined at the left boundary: $c(L)\le r(L)$. The valid columns form a prefix, whose last index is again nondecreasing. All interval endpoints can therefore be obtained by merging the boundary orders in time linear in the list sizes.

\emph{Anti-Monge inequality.}
Take two auxiliary columns $c,d$, with $c$ below $d$, and two right-carrier rows $r,r'$, with $r$ below $r'$, and suppose all four matrix entries are valid. The boundary-incidence argument above ensures that both auxiliary carriers reach the left boundary and both right carriers reach the right boundary of the strip. Lemma~\ref{lem:balanced-support-difference} therefore applies and gives $|zc_r|\le |zc_{r'}|$, and $|zd_r|\le |zd_{r'}|$.

In~\textup{(B5e)}, the subtracted triangle $\triangle z p_c p_d$ is independent of the right carrier. The other triangle has area $\frac12 |zc_r|\,|zd_r|\sin\phi$, where $\phi$ is the fixed angle between $\tau_c$ and $\tau_d$. Consequently,
\[
        F_{r'}(d)-F_{r'}(c)
        \ge F_r(d)-F_r(c).                                \tag{B5g}
\]
Thus every fully defined submatrix satisfies the anti-Monge inequality. Notice that the row-dependent term $W_r$ cancels automatically; no ordering of the auxiliary points $t_r$ is required.

\emph{Padding and SMAWK.}
After deleting rows with no valid entry, let $[\ell_j,u_j]$ denote the valid column interval of row $j$, with a prefix or suffix represented in the same way. Pad the undefined entries symbolically as follows. Every padded value is smaller than every valid area; left-padding values increase with the column index, right-padding values decrease with the column index, and every left-padding value is larger than every right-padding value. Let $M$ denote the resulting matrix.

We claim that $M$ is totally monotone for row maxima. Thus, whenever $j<j'$ and $k<k'$,
\[
        M[j,k]\le M[j,k']
        \quad\Longrightarrow\quad
        M[j',k]\le M[j',k'].                             \tag{B5h}
\]
Indeed, if $k'>u_j$, the entry at $k'$ is right padding, and the antecedent is false. If $k'<\ell_j$, both entries in the upper row are left padding; since $\ell_{j'}\ge\ell_j$, the same increasing padding order holds in the lower row. Otherwise $k'$ is valid in row $j$. Because $u_{j'}\ge u_j$, it is not right padding in row $j'$. If $k<\ell_{j'}$, the lower entry at $k$ is left padding and is smaller than the entry at $k'$. In the remaining case all four entries are valid, and~\textup{(B5g)} gives~\textup{(B5h)}.

Therefore, SMAWK~\cite{AggarwalKMSW87} finds all row maxima using $O(\#\text{rows}+\#\text{columns})$ entry evaluations. Each valid entry and its corresponding quadrilateral witness are computed in constant time, and the matrices need not be materialized.

\emph{Time and space.}
At any fixed level of the segment tree, each carrier has at most two standard and at most two hereditary occurrences. Hence the total list size at one level is $O(n)$, and over all $O(\log n)$ levels the total construction and matrix-search time is $O(n\log n)$.

To keep the working space linear, generate and process the lists one tree level at a time. The current and next levels contain only $O(n)$ records in total; after a matrix is searched, retain only its best candidate and discard its lists. Thus the working space is $O(n)$.~\hfill$\qed$
\end{proof}

Considering $O(n)$ possible balanced centres and a symmetric negative-slope case, we have the following.
\begin{lemma}\label{lem:balanced-branch}
All non-exposed balanced--carrier optima with nonzero top slope that are not already covered by the carrier--carrier branch can be found in $O(n^2\log n)$ time and $O(n)$ working space.
\end{lemma}

\begin{proof}
Fix a balanced centre $v_i$. Its feasible balanced-endpoint intervals and the retained auxiliary family $\mathcal A_i$ are constructed in $O(n)$ time and $O(n)$ space.

If the balanced left side of an optimum is also carrier-representable, the configuration has already been considered by the carrier--carrier branch. Otherwise, Lemma~\ref{lem:balanced-top-carriers} represents every remaining non-exposed positive-slope optimum by an intersection of a carrier in $\mathcal A_i$ with an original right carrier. Lemma~\ref{lem:balanced-matrix-search} finds the best such intersection in $O(n\log n)$ time and $O(n)$ working space.

We process the $O(n)$ possible balanced centres successively. The negative-slope case is symmetric. Hence the total running time is $O(n^2\log n)$, while the working space remains $O(n)$.~\hfill$\qed$
\end{proof}

\subsubsection{Exposed-corner configurations}\label{subsec:corner-active}

We now turn to configurations with an exposed upper corner. Such a corner is
handled outside the interior carrier sweep even when its incident side is
also carrier-representable, because the endpoint itself can move along the
upper terrain chain. We write a grounded quadrilateral as
$A,\beta,\gamma,D$, where $\beta$ and $\gamma$ are the left and right
upper endpoints.

A recurring subproblem is to determine how far the corresponding base
endpoint can move while the side remains inside the terrain. A base point
$b$ is \emph{visible} from an upper-chain point $p$ if
$\overline{bp}\subseteq\mathfrak T$. For $p\in\Gamma$, let
$b_\ell(p)$ and $b_r(p)$ be the leftmost and rightmost visible base points,
respectively. As $p$ moves along $\Gamma$, the active terrain contact that
limits visibility changes only at discrete events. We therefore subdivide
$\Gamma$ into maximal intervals on which this contact, or the
direct-visibility state, remains fixed. These subdivisions are the
\emph{left-base} and \emph{right-base visibility maps}. The next lemma
establishes their size and the algebraic form of the visible endpoint within
one cell.

\begin{lemma}\label{lem:kinetic-base-visibility}
The left- and right-base visibility maps can each be constructed in $O(n)$ time and $O(n)$ space, and each has $O(n)$ cells. On every cell, parameterize the moving upper-chain point on its current terrain edge $\overline{vw}$ by $p(t)=(1-t)v+tw,\ 0\le t\le1$. Then the coordinate of the corresponding extreme visible base point is either constant or of the form $\frac{a_0+a_1t}{b_0+b_1t}$, where $a_0,a_1,b_0,b_1$ are constants for that cell and the denominator is nonzero throughout the cell. We call such a function \emph{linear-fractional}.
\end{lemma}

\begin{proof}
We prove the statement for the right-base visibility map; the left-base map
follows by horizontal reflection. For the right-base map the upper chain is
traversed from right to left, so the completed terrain vertices form exactly
the suffix relevant to visibility from the current point.

Let $\overline{vw}$ be the current terrain edge, with $v$ to the left of $w$.
During the right-to-left traversal we use the local parameter $p(\tau)=(1-\tau)w+\tau v$, for $0\le \tau\le1$.
Let $C$ be the fixed set consisting of the completed terrain vertices to the
right of the current edge together with $w$. For
$0\le \tau<\tau'\le1$ with $\tau'>0$,
\[
        p(\tau)
        =
        \left(1-\frac{\tau}{\tau'}\right)w
        +
        \frac{\tau}{\tau'}p(\tau').
\]
Hence $p(\tau)\in\overline{wp(\tau')}$, and therefore $\operatorname{conv}\bigl(C\cup\{p(\tau)\}\bigr) \subseteq \operatorname{conv}\bigl(C\cup\{p(\tau')\}\bigr)$. Thus the maintained convex hulls are nested while the sweep crosses one terrain edge.

After horizontal reflection, this is exactly the usual left-to-right Andrew
lower-hull sweep. Consequently, a completed vertex removed from the lower
hull cannot become a strict lower-hull vertex again. Intermediate collinear
vertices are omitted by the non-strict orientation rule. The orientation
predicate involving the current last hull edge and the moving point is affine
in $\tau$, so every change of the active lower tangent is a suffix-pop event.
At $\tau=1$, the endpoint $v$ becomes completed and the next terrain edge is
processed with a new local parameter. Every terrain vertex is inserted once
and, after it is popped, is never reinserted. Hence the total number of hull
updates and active-tangent changes is $O(n)$.

For $p=(p_x,p_y)\in\Gamma$, if $\overline{pB_r}$ is feasible, then $B_r$
is the rightmost visible base point. Each maximal interval of this
direct-visibility state forms one visibility cell. Otherwise, the blocking
contacts lie in the terrain suffix strictly to the right of $p$. We define
\[
        m^*(p)=
        \min_{\substack{z\in V(\mathfrak T)\\ z_x>p_x}}
        \frac{z_y-p_y}{z_x-p_x},
\]
where $B_r$ is included among the candidate vertices. Thus $m^*(p)$ is the slope of the lower supporting line from $p$ to the terrain suffix. Let $\omega$ be the first vertex attaining this minimum in increasing $x$-order. The line $p\omega$ lies on or below every intervening terrain vertex and hence below every intervening terrain edge. Its intersection with the base has abscissa $d^*(p)=p_x-\frac{p_y}{m^*(p)}$. The segment from $p$ to $(d^*(p),0)$ is feasible. Any farther base point determines a larger slope and passes strictly above $\omega$, so $(d^*(p),0)$ is the rightmost visible base point.

Because the right-to-left sweep maintains the lower hull of precisely this
completed suffix, $\omega$ is exactly its active lower-tangent contact. While
$\omega=(x_\omega,y_\omega)$ is fixed and $p(\tau)=(x(\tau),y(\tau))$
moves affinely on one terrain edge, the supporting line through
$p(\tau)$ and $\omega$ meets the base $y=0$ at
\[
        d(\tau)=
        \frac{x(\tau)y_\omega-y(\tau)x_\omega}
             {y_\omega-y(\tau)}.
\]
Both numerator and denominator are affine in $\tau$. Hence, on every cell where the denominator is nonzero,
\[
        d(\tau)=\frac{a_0+a_1\tau}{b_0+b_1\tau}
\]
for constants determined by that cell.

The sweep is subdivided at terrain-edge transitions, active-tangent changes, direct-visibility transitions, and zeros of the selected denominator. Terrain-edge events and active-tangent changes occur $O(n)$ times, and direct-visibility transitions are charged to the same tangent sequence. Thus the subdivision has $O(n)$ cells. Boundary values are evaluated explicitly. The parameter $\tau$ differs from the left-to-right edge parameter in the lemma statement only by the affine change $t=1-\tau$, which preserves the linear-fractional form.

Thus the right-base visibility map has $O(n)$ cells, is constructed in $O(n)$ time, and uses $O(n)$ space. Horizontal reflection gives the left-base visibility map with the same bounds.~\hfill$\qed$
\end{proof}

We first consider configurations with exactly one exposed upper corner. For a positive-slope top, the possibilities are: right-exposed with a carrier or balanced left side, and left-exposed with a carrier right side. The negative-slope cases follow by reflection; the horizontal balanced--exposed case is deferred to Lemma~\ref{lem:quadratic-horizontal}.

\begin{lemma}\label{lem:moving-corner-envelope}
For each fixed non-exposed carrier, or fixed balanced centre with nonzero top slope, a one-exposed sweep has an $O(n)$-size event subdivision. On each event interval, the feasible positions on the fixed support form an interval whose endpoints and relevant extreme visible base coordinate are fixed or linear-fractional in the moving-corner parameter. Hence the sweep is optimized in $O(n)$ time and space, with interior stationary points determined by equations of degree at most four. Carrier cases include their horizontal boundary events.
\end{lemma}

\begin{proof}

We first consider the positive-slope cases in which the right upper corner $\gamma$ is exposed. The non-exposed left side is either a carrier or balanced. Their negative-slope counterparts follow by horizontal reflection. The remaining positive-slope case, with exposed left corner and a non-exposed right carrier, is treated separately below.

For the two right-exposed cases, the opposite upper endpoint $\beta$ lies on a fixed line $L$: either the supporting line of the opposite carrier or $H_i$ for a balanced centre $v_i=(x_i,y_i)$ with $h_i=2y_i$. Let $K(\gamma)\subseteq L$ denote the positions of $\beta$ satisfying the top-side, base-side, and convex-support conditions.

For fixed $\beta$ and $\gamma$, let $D=(d,0)$ be the right base endpoint. The shoelace formula is affine in $d$, with coefficient $\gamma_y/2>0$. Thus the area increases as $D$ moves rightward while feasibility is preserved. In the clockwise order $A,\beta,\gamma,D$, $\operatorname{orient}(\beta,\gamma,D)$ is nonincreasing in $d$ for a nonnegative top slope, while
\[
 \operatorname{orient}(\gamma,D,A)=-\gamma_y(d-A_x)<0,  \qquad \operatorname{orient}(D,A,\beta)=(A_x-d)\beta_y<0
\]
for $d>A_x$; the orientation at $\beta$ is independent of $d$. Hence, once a visible position of $D$ is feasible, every farther visible position remains feasible. An optimum therefore uses the rightmost base point visible from $\gamma$. By Lemma~\ref{lem:kinetic-base-visibility}, its abscissa is constant or linear-fractional on each visibility cell.

\smallskip\noindent\emph{Carrier representation.} Let $\ell$ be the fixed carrier of the non-exposed side, $p=q(\ell)$ its terminal contact, and $J$ the admissible portion of $\ell$
containing $\beta$. By Lemma~\ref{lem:top-frontier}, the carrier shoulder from $\beta$ to $p$ is feasible, so only the upper-chain subarc from $p$ to $\gamma$ can obstruct the top side.

For fixed $\gamma$, let $k$ be the slope of $\overline{\beta\gamma}$. Every relevant terrain vertex $z$ between $p$ and $\gamma$, together with $p$, imposes
\[
        k\ge\frac{z_y-\gamma_y}{z_x-\gamma_x}.               \tag{E1}
\]
Thus top feasibility is equivalent to $k\ge\lambda(\gamma)$, where $\lambda(\gamma)=\max_z\frac{z_y-\gamma_y}{z_x-\gamma_x}$. As $k$ varies without crossing the slope of $\ell$, the intersection of the line through $\gamma$ with $\operatorname{line}(\ell)$ moves monotonically along $\ell$. Hence~\textup{(E1)} selects an interval of $J$. The positive-top-slope and convex-support conditions also select half-intervals, because their orientation determinants are affine in the scalar position of $\beta$ on $\ell$. Their intersection is therefore an interval $K(\gamma)$. As $\gamma$ moves from left to right, the vertex attaining $\lambda(\gamma)$ is the active lower-tangent contact. The suffix-pop argument of Lemma~\ref{lem:kinetic-base-visibility} applies: each completed terrain vertex enters the maintained lower hull once and, after removal, never returns. Hence the active contact changes only $O(n)$ times.

Between two consecutive tangent changes, let $z$ be the fixed active contact and let $\gamma$ move on one terrain edge. The terrain-feasibility boundary on $J$ is the intersection of $\operatorname{line}(\gamma z)$ with the fixed carrier line. Since $\gamma$ is affine in the edge parameter, this intersection is linear-fractional.

The other candidate boundaries of $K(\gamma)$ have the same form. The endpoints of $J$ are fixed; the zero-top-slope boundary comes from the horizontal line through $\gamma$; and the support-order boundary at $\beta$ comes from $\operatorname{line}(A\gamma)$. On one base-visibility cell, $\operatorname{line}(\gamma D)$ passes either through $B_r$ or through the
fixed active visibility contact, so its intersection with $\ell$ is also linear-fractional. Thus every candidate endpoint of $K(\gamma)$ is fixed or linear-fractional on the current cell.

The terrain-edge, tangent, visibility, carrier-parallel, and support-boundary events give $O(n)$ initial intervals. On one such interval, two candidate endpoint functions
\[
        f(t)=\frac{a_0+a_1t}{b_0+b_1t},        \qquad        g(t)=\frac{c_0+c_1t}{d_0+d_1t},
\]
with nonzero denominators, can exchange order only at roots of $(a_0+a_1t)(d_0+d_1t) = (c_0+c_1t)(b_0+b_1t)$, unless they agree identically. This equation has degree at most two.
Since only constantly many candidate boundaries occur on one initial interval, their crossings add only $O(1)$ breakpoints. Hence the carrier sweep has $O(n)$ event intervals, and on each of them $K(\gamma)$ is an interval with fixed or linear-fractional endpoints.

\smallskip\noindent\emph{Balanced representation.} Fix a balanced centre $v_i=(x_i,y_i)$ and let $H_i$ be the line $y=h_i$, where $h_i=2y_i$. We consider $\gamma_y>h_i$; the horizontal case is handled in Lemma~\ref{lem:quadratic-horizontal}. If $\beta=(\beta_x,h_i)$ and $\overline{\beta\gamma}$ has positive slope, then every point of the top segment to the right of $\beta$ has height greater than $h_i$. Feasibility therefore implies $T(x)>h_i$ for $(\beta_x<x\le\gamma_x)$. Thus $\beta$ and $\gamma$ lie over one height-$h_i$ sweep component
$I=[a,b]$. Let $J_I\subseteq H_i$ be its feasible balanced-endpoint interval.

For $\beta=(x,h_i)\in J_I$ and $\beta_x<z_x<\gamma_x$, the top height at $z_x$ is $h_i+(\gamma_y-h_i)\frac{z_x-x}{\gamma_x-x}$. Its derivative with respect to $x$ is nonpositive, so moving $\beta$ to the right cannot raise the top over the common part of its projection. Consequently, the terrain-feasible subset of $J_I$ is a suffix. Its left endpoint is determined by the lower supporting line from $\gamma$ to the relevant lower hull. If $z=(z_x,z_y)$ is the active tangent contact, that line meets $H_i$ at $z_x+(h_i-z_y)\frac{\gamma_x-z_x}{\gamma_y-z_y}$, which is linear-fractional while $z$ is fixed and $\gamma$ moves on one terrain edge. The balanced support-order condition is affine in the coordinate of $\beta$ and therefore selects a half-interval. Intersecting it with the terrain-feasible suffix and $J_I$ gives $K(\gamma)$.

The event count is the same as in the carrier case. Tangent changes are suffix-pop events, and the disjoint sweep components process $O(n)$ terrain vertices in total for the fixed centre. The right-base visibility map adds $O(n)$ boundaries. All remaining candidate endpoints are fixed or linear-fractional, and the crossing argument above adds only $O(1)$
breakpoints per initial interval. Hence a fixed balanced centre also gives  $O(n)$ event intervals, with $K(\gamma)$ having fixed or linear-fractional endpoints on each.

\smallskip
\noindent
\emph{Optimization on one event interval.}
Parameterize the exposed corner along one terrain edge as
\[
        \gamma(t)=g_0+t(g_1-g_0),  \qquad 0\le t\le1.
\]
For fixed $\gamma$, the area is affine in the position of $\beta$. In the carrier case, $A$ is fixed and $\beta$ moves on a fixed line. In the balanced case, $A=(2x_i-\beta_x,0)$, for $\beta_y=h_i$, and substitution in the shoelace formula again gives an affine function of $\beta_x$. Therefore the maximum over $K(\gamma)$ occurs at one of its two endpoints.

On a fixed event interval, an endpoint of $K(\gamma(t))$ has coordinates
\[
        \beta_x(t)=\frac{N_x(t)}{D_\beta(t)},        \qquad        \beta_y(t)=\frac{N_y(t)}{D_\beta(t)},
\]
where $N_x,N_y,D_\beta$ are affine in $t$; constant coordinates correspond to denominator $1$. The selected right-base coordinate has the form $d(t)=\frac{N_d(t)}{D_d(t)}$ by Lemma~\ref{lem:kinetic-base-visibility}.

For the clockwise quadrilateral $A,\beta,\gamma,D$, $\area(A\beta\gamma D) = \frac12\left(\beta_y(\gamma_x-A_x)+\gamma_y(d-\beta_x)\right)$. Here $A_x$ is constant in the carrier case and $A_x=2x_i-\beta_x$ in the balanced case. Substitution gives 
\[
        \Phi(t)=\frac{P(t)}{Q(t)},        \qquad        Q(t)=D_\beta(t)D_d(t),        \qquad        \deg P\le3,\quad \deg Q\le2.
\]
Zeros of the active denominators and parallel-support values are event boundaries, so $Q$ has fixed nonzero sign in the interior of an event interval. Moreover,
\[
        \Phi'(t)    =   \frac{P'(t)Q(t)-P(t)Q'(t)}{Q(t)^2},
\]
whose numerator has degree at most four. Unless it vanishes identically, all interior stationary points are therefore roots of a polynomial of degree at most four. Together with the interval boundaries, these contain all maxima. If the numerator vanishes identically, $\Phi$ is constant, and the maximum-level preference test selects a nondegenerate optimum whenever one exists.

Each event interval has two endpoint functions and requires $O(1)$ work. Thus each right-exposed sweep takes $O(n)$ time and $O(n)$ working space.

\smallskip
\noindent
\emph{Left-exposed carrier representation.}
It remains to consider the positive-slope case with exposed left corner$\beta$ and a non-exposed right side on a fixed carrier $r$. Let $J$ be the admissible portion of $r$. For fixed $\beta$, if $k$ is the slope of $\overline{\beta\gamma}$, top feasibility is equivalent to
\[
        k\le\mu(\beta),        \qquad        \mu(\beta)=        \min_z\frac{z_y-\beta_y}{z_x-\beta_x},           \tag{E3}
\]
where $z$ ranges over the relevant upper-chain vertices between $\beta$ and the terminal contact of $r$; if there is no intervening subchain, this constraint is absent.

We sweep $\beta$ from right to left and maintain the minimizing contact by the same kinetic lower-hull mechanism, with the sweep direction reversed. While this contact and the terrain edge containing $\beta$ remain fixed, the terrain-feasibility boundary on $J$ is linear-fractional. The positive-slope and convex-support conditions again select half-intervals,
and Lemma~\ref{lem:kinetic-base-visibility} makes the leftmost visible base coordinate fixed or linear-fractional. The same event-merging and crossing argument therefore gives $O(n)$ event intervals.

For fixed $\beta$, the area is affine in the position of $\gamma$ on $r$, so an optimum occurs at an endpoint of the feasible interval. Moving the left base endpoint outward either reaches the leftmost visible base point or first makes $\operatorname{orient}(A,\beta,\gamma)=0$, which is a triangular boundary candidate covered by $\Delta^*$. Hence no nontriangular optimum is lost by taking the extreme visible base point.

Substitution of the resulting linear-fractional endpoint functions into the area formula gives the same degree bounds as above: numerator degree at most three, denominator degree at most two, and stationary equations of degree at most four. Thus one fixed right carrier is processed in $O(n)$ time and $O(n)$ working space.

The three positive-slope one-exposed cases are therefore covered. Reflection gives the negative-slope cases, and the horizontal boundaries are handled as stated in the lemma.~\hfill$\qed$
\end{proof}

We next consider configurations in which both upper corners are exposed. Here both upper endpoints vary along the upper chain, so the previous one-parameter subdivision is replaced by a two-parameter event-cell decomposition. 

\begin{lemma}\label{lem:two-exposed}
The doubly-exposed branch, including its horizontal boundary events, runs in
$O(n^2)$ time and uses $O(n)$ working space.
\end{lemma}

\begin{proof}
We analyze the case in which the top has nonnegative slope; the negative-slope case follows by reflection. We first fix a cell $\Sigma$ of the left-base visibility map, refined at terrain vertices, and assume that $\Sigma$ is contained in a terrain edge $e$. There are $O(n)$ such cells. We parameterize the exposed left upper corner on $\Sigma$ by an affine function $\beta(s)$. By Lemma~\ref{lem:kinetic-base-visibility}, the leftmost visible base coordinate $a(s)$ is constant or linear-fractional throughout $\Sigma$.

We next consider the position of the exposed right upper corner $\gamma(t)$. If both exposed corners lie on $e$, the top segment is contained in $e$ and there is no intervening-terrain constraint. For the fixed left cell $\Sigma$, we subdivide the admissible positions of $\gamma$ on $e$ by the right-base visibility boundaries and by the zeros of the rational
denominators used below. The right-base visibility map has $O(n)$ cells, and only constantly many affine denominators are active on each such cell. Hence this same-edge case contributes $O(n)$ constant-complexity cells for the fixed $\Sigma$. On each cell the remaining base-order and convex-support conditions, together with the area function, have exactly the same
bounded-degree rational form used below; the top-terrain inequality is simply absent. Thus these cells are handled by the same algebraic optimization.

For the remainder of the proof, assume that $\gamma(t)$ lies strictly to the right of $e$. For a fixed position of $\gamma$, let $k$ be the slope of the line through $\gamma$ and $\beta$. Every terrain vertex $z$ strictly between the two corners satisfies $\beta_x<z_x<\gamma_x$, and feasibility of the top segment requires $k\ge\frac{z_y-\gamma_y}{z_x-\gamma_x}$. Thus all intervening-terrain constraints have the same orientation in the slope parameter, and top feasibility is equivalent to $k\ge\lambda(t)$, where $\lambda(t)=\max_z\frac{z_y-\gamma_y(t)}{z_x-\gamma_x(t)}$, where the maximum ranges over the intervening terrain vertices.

For fixed $\gamma$, the intersection of a line through $\gamma$ with the fixed edge $e$ varies monotonically with its slope, except at the direction parallel to $e$. Hence the feasible positions of $\beta\in\Sigma$ form an interval. Its terrain-defined boundary is the intersection of $e$ with the lower supporting line from $\gamma$ through the vertex attaining
$\lambda(t)$.

For the fixed cell $\Sigma\subset e$, we consider the kinetic lower hull of the right endpoint of $e$, the completed terrain vertices between $e$ and the current position of $\gamma$, and the moving point $\gamma(t)$. The hull edge incident to $\gamma(t)$ determines the lower supporting line above. Hence the nested-hull argument of Lemma~\ref{lem:kinetic-base-visibility} applies directly. Every completed terrain vertex is inserted once, and a vertex removed by a suffix pop never returns to the lower hull. Therefore the active tangent changes only $O(n)$ times for the fixed cell $\Sigma$.

The value $\lambda(t)$ is independent of the parameter $s$ of $\beta$. Indeed, the relative interior of $\Sigma$ lies on a single open terrain edge. Varying $\beta$ within $\Sigma$ crosses no terrain vertex and therefore does not change the set of strictly intervening vertices. The endpoints of $\Sigma$ are retained as closed events. Consequently, one active tangent
identity is valid for all $\beta\in\Sigma$ throughout each tangent interval.

For the fixed cell $\Sigma$, the parameter domain of $\gamma$ has only $O(n)$ relevant event values. These arise from terrain-edge endpoints, active-tangent changes, right-base visibility boundaries, and zeros of the constantly many rational denominators active on the current interval. The first three families have linear total size. Between consecutive such
events, each active denominator is affine and contributes at most one additional zero. Hence the resulting subdivision contains $O(n)$ event intervals for the fixed cell $\Sigma$.

On each event interval, $\gamma(t)$ is affine in $t$, the active tangent vertex is fixed, and the rightmost visible base coordinate $d(t)$ is constant or linear-fractional. The terrain-feasibility boundary on $\Sigma$ is likewise fixed or linear-fractional, since it is obtained by intersecting the fixed edge $e$ with the line through $\gamma(t)$ and the
fixed active tangent vertex.

We now describe the feasible set on one resulting two-parameter geometric event cell. The extreme visible base coordinates have the form
\[
        a(s)=\frac{N_a(s)}{D_a(s)},        \qquad        d(t)=\frac{N_d(t)}{D_d(t)},
\]
where the numerators and denominators are affine; a constant visibility coordinate is represented with denominator $1$. We write $A(s)=(a(s),0)$, $D(t)=(d(t),0)$, and denote the parameter ranges of the cell by $s_0\le s\le s_1$ and $t_0\le t\le t_1$. The denominators $D_a$ and $D_d$ have fixed nonzero signs in the interior of the cell. Accordingly, we use
\[
        \sigma_T=1,\qquad        \sigma_L=-\operatorname{sgn}(D_a),\qquad        \sigma_R=-\operatorname{sgn}(D_d),\qquad        \sigma_B=\operatorname{sgn}(D_aD_d).
\]
With these signs, the feasibility conditions are represented by the polynomial inequalities in Table~\ref{tab:exposed-cell}.

\begin{table}[t]
\centering
\caption{Polynomial tests on one doubly-exposed event cell.}
\label{tab:exposed-cell}
\small
\renewcommand{\arraystretch}{1.12}
\begin{tabularx}{\linewidth}{@{}>{\raggedright\arraybackslash}X
                >{\centering\arraybackslash}p{.10\linewidth}
                >{\raggedright\arraybackslash}p{.25\linewidth}@{}}
\hline
Constraint and sign-adjusted polynomial test & Degree & Boundary event\\
\hline
Top below the terrain:\newline $\sigma_T\operatorname{orient}(\beta(s),\gamma(t),z)\ge0$ for the fixed active tangent vertex $z$ & $\le2$ & tangent change or terrain vertex\\

Corner order: $\gamma_x(t)-\beta_x(s)\ge0$ & $\le1$ & coincident vertical line\\

Top-slope branch: $\gamma_y(t)-\beta_y(s)\ge0$ & $\le1$ & horizontal top\\

Base order:\newline $\sigma_B\bigl(N_d(t)D_a(s)-N_a(s)D_d(t)\bigr)\ge0$ & $\le2$ & coincident base endpoints\\

Left support order:\newline $\sigma_LD_a(s) \operatorname{orient}(A(s),\beta(s),\gamma(t))\ge0$ & $\le3$ & zero orientation\\

Right support order:\newline $\sigma_RD_d(t) \operatorname{orient}(\beta(s),\gamma(t),D(t))\ge0$ & $\le3$ & zero orientation\\

Parameter ranges: $s_0\le s\le s_1$ and $t_0\le t\le t_1$ & $1$ & terrain or visibility endpoint\\ \hline
\end{tabularx}
\end{table}

The first inequality is equivalent to $k(s,t)\ge\lambda(t)$, because the corner-order condition fixes the signs of the ray denominators. The base-order and the two support-order conditions are the only inequalities involving the rational expressions for $a(s)$ and $d(t)$. Hence the feasible portion of the geometric event cell is described by a constant number of polynomial inequalities of degree at most three, apart from the singular boundaries of the chosen rational charts.

The event decomposition includes every zero of $D_a$, every zero of $D_d$, and every zero of the denominator arising from the intersection of the active supporting line with $e$. Thus no open event cell contains a singular value of its chosen rational representation. At a singular boundary, if a visibility coordinate becomes unbounded, the bounded terrain
base excludes a neighborhood of that parameter value from the feasible set. If the corresponding geometric configuration has a finite limit, that boundary configuration is evaluated directly from its defining lines or represented in an adjacent nonsingular chart.

Accordingly, we partition the feasible portion of each geometric event cell into a constant number of compact semialgebraic charts, treating the finite singular boundary configurations separately as above. On every remaining chart, all active denominators stay away from zero and therefore have fixed nonzero sign on its closure. The conditions in
Table~\ref{tab:exposed-cell} then give each chart constant description complexity, so Fact~\ref{res:algebraic-cell} applies.

We finally determine the relevant base endpoints. For fixed top corners, moving the left base endpoint to the left and the right base endpoint to the right strictly increases the area whenever the corresponding side has positive height. In the present nonnegative-slope branch, rightward motion of $D$ preserves convexity by the orientation argument of
Lemma~\ref{lem:moving-corner-envelope}.

For the left base endpoint $A=(a,0)$, the only support condition that can become tight under leftward motion is $\operatorname{orient}(A,\beta,\gamma)<0$. Its derivative with respect to
$a$ is $\beta_y-\gamma_y\le0$. Hence loss of convexity can first occur when $\operatorname{orient}(A,\beta,\gamma)=0$. At this boundary, the vertices $A,\beta,\gamma$ are collinear. The
corresponding closed configuration therefore has triangular convex hull and area at most $\area(\Delta^*)$. Thus every nontriangular optimum for fixed top corners uses both extreme visible base points. If the extreme left visible endpoint violates the support-order condition, the intervening boundary is precisely this triangular event, which is retained among the
boundary candidates. The negative-slope case follows by reflection.

For the resulting nontriangular candidate, $A=(a(s),0)$ and $D=(d(t),0)$, and the area is
\[
        \Phi(s,t)=\frac12\left(\beta_y(s)[\gamma_x(t)-a(s)]+\gamma_y(t)[d(t)-\beta_x(s)]\right).
\]
Substituting $a(s)=\frac{N_a(s)}{D_a(s)}$, and $d(t)=\frac{N_d(t)}{D_d(t)}$ gives
\[
        \Phi(s,t)=\frac{P(s,t)}{Q(s,t)},  \qquad   \deg P\le4,\quad \deg Q\le2.
\]
On each nonsingular chart, $Q$ has fixed nonzero sign. By Fact~\ref{res:algebraic-cell}, its maximum is therefore obtained exactly in $O(1)$ time. The maximum-level preference rule from the sweep conventions resolves equal optimum values in favor of a nondegenerate maximizer whenever one exists.

For each fixed left visibility cell $\Sigma$, the strictly-right case has $O(n)$ event intervals, while the same-edge case contributes $O(n)$ additional constant-complexity cells. Each event cell yields only a constant number of nonsingular charts and requires $O(1)$ total optimization time. Since the left-base visibility map has $O(n)$ cells, the total running time
is $O(n^2)$. The kinetic hull and event data for one cell $\Sigma$ occupy $O(n)$ space and are discarded after that cell is processed. Hence the working-space bound is $O(n)$. Equality of the two top heights is retained as a closed horizontal boundary event.~\hfill$\qed$
\end{proof}



Lemmas~\ref{lem:moving-corner-envelope} and~\ref{lem:two-exposed} therefore handle all nonzero-slope exposed configurations in $O(n^2)$ time and $O(n)$ space. The former also includes carrier--exposed horizontal events; the remaining balanced--exposed horizontal cases are covered below.

\subsubsection{Horizontal top edges}\label{subsec:horizontal}

Horizontal tops have already appeared as boundary events in the carrier and
exposed-corner branches. A separate treatment is needed only when at least
one side is balanced. In this situation both upper endpoints lie on the
same height level, and containment of the top edge reduces to membership in
a common closed level component. This makes the remaining compatibility
conditions considerably simpler than in the nonzero-slope case.

\begin{lemma}\label{lem:quadratic-horizontal}
All remaining horizontal-top configurations, including the balanced--exposed, balanced--carrier, and double-balanced cases, can be processed in $O(n^2)$ additional time and $O(n)$ working space under the closed convention. 
\end{lemma}

\begin{proof}
Lemma~\ref{lem:carrier-branch} already covers the non-exposed carrier--carrier case at zero top slope. The carrier--exposed horizontal case is included as a boundary case of Lemma~\ref{lem:moving-corner-envelope}, and Lemma~\ref{lem:two-exposed} covers configurations with two exposed upper corners. The balanced representation in
Lemma~\ref{lem:moving-corner-envelope} uses the strict inequality $g_y>h_i$, so the horizontal balanced--exposed case is handled here together with the balanced--carrier and double-balanced cases.

We first consider a balanced left side with centre $v_i=(x_i,y_i)$. Let $h_i=2y_i$ and let $H_i$ denote the horizontal line at height $h_i$. Lemma~\ref{lem:balanced-interval-active} yields the feasible endpoint intervals $J_{i,\nu}=[\lambda_{i,\nu},U_{i,\nu}]$ on the closed level components $I_{i,\nu}\subseteq\{x:T(x)\ge h_i\}$.

\noindent
\emph{Balanced--exposed horizontal configurations.}
We consider an exposed right upper corner $\gamma=(g_x,h_i)$ belonging to the same closed level component $I_{i,\nu}$ as the balanced top endpoint. Let $d(\gamma)$ be the abscissa
of the rightmost base point visible from $\gamma$. A positive-length horizontal top with balanced left endpoint $p=(p_x,h_i)$ exists exactly when
\[
        J_{i,\nu}
        \cap
        \bigl(2x_i-d(\gamma),\,g_x\bigr)
        \ne\varnothing.                                    \tag{H0e}
\]
Indeed, $p_x<g_x$ is the strict top-endpoint order, while the balanced base endpoint is $(2x_i-p_x,0)$, so strict base order is equivalent to $p_x>2x_i-d(\gamma)$. Since $p$ and $\gamma$ lie in the same closed level component of $\{x:T(x)\ge h_i\}$, the horizontal segment $\overline{p\gamma}$ is contained in the terrain. Thus~\textup{(H0e)} is
both necessary and sufficient.

For every feasible $p_x$, the area is $\frac{h_i}{2}\bigl(d(\gamma)+g_x-2x_i\bigr)$, and is independent of the particular choice of $p_x$. Notice that the use of the closed level component is essential: the horizontal top may pass through several strict height-$h_i$ sweep components separated by points or segments where $T(x)=h_i$.

For the fixed centre $v_i$, we scan the right-base visibility map together with the level components at height $h_i$. On a nonhorizontal terrain edge, $H_i$ meets the edge in at most one exposed position. On a horizontal terrain edge contained in $H_i$, the restriction of $d(\gamma)$ to each visibility cell is constant or linear-fractional (Lemma~\ref{lem:kinetic-base-visibility}). For $J_{i,\nu}=[\lambda_{i,\nu},U_{i,\nu}]$, condition~\textup{(H0e)} gives $\lambda_{i,\nu}<g_x$ and $U_{i,\nu}>2x_i-d(\gamma)$. On one visibility cell these inequalities have only constantly many boundary values, and the area expression above is a fixed-degree rational function of the cell parameter. Hence each cell is processed in $O(1)$ time. The
right-base visibility map has $O(n)$ cells, so all balanced-left--exposed horizontal configurations for $v_i$ require $O(n)$ time. Over all balanced centres this contributes $O(n^2)$ time, $O(n)$ space. The balanced-right--exposed-left case follows by reflection.

\noindent
\emph{Balanced--carrier horizontal configurations.}
Among the right carriers, only those reaching height $h_i$ are relevant. For such a carrier $r$, let $t_r=r\cap H_i$, and $b_r=b(r)$, where $t_r$ lies on the stored carrier segment and $b_r$ is the intersection of its supporting line with the base. If the terminal contact of $r$ lies at height $h_i$, then the top-right endpoint is exposed; this configuration is already covered by the balanced--exposed horizontal case above. Hence only carriers whose terminal contact $q(r)$ lies strictly above $H_i$ remain.

For every such carrier, the relative interior of the segment $\overline{t_rq(r)}$ lies strictly above $H_i$. Feasibility therefore gives $T(x)>h_i$
throughout its horizontal projection, so $\overline{t_rq(r)}$ belongs to a unique height-$h_i$ sweep component. Let $\Omega$ be the region bounded below
by the corresponding level-$h_i$ segment, above by the associated upper-chain subarc, and laterally by the two vertical component boundaries. Each segment
$\overline{t_rq(r)}$ is a cross-cut of $\Omega$. By Fact~\ref{lem:carrier-families}, the interiors of these cross-cuts are pairwise disjoint. Lemma~\ref{lem:crosscut-order} therefore implies that the points $t_r$ occur along the level-$h_i$ boundary in the same order as their terminal contacts $q(r)$ occur along the upper-chain boundary. Consequently,
the boundary order of the terminal contacts and the sweep-component order determine, in linear time, both the component containing each $t_r$ and the order of the $t_r$-points within that component. Closed splitting events are handled in the same event batch.

Suppose that $t_r$ belongs to $I_{i,\nu}$. A positive-length horizontal top exists when
\[
        J_{i,\nu}
        \cap
        \bigl(2x_i-(b_r)_x,\,(t_r)_x\bigr)
        \ne\varnothing.                                     \tag{H0}
\]
Indeed, if $p=(p_x,h_i)$ is the balanced top endpoint, the strict top order is $p_x<(t_r)_x$, while the balanced base endpoint is $(2x_i-p_x,0)$, $p_x>2x_i-(b_r)_x$. Thus~\textup{(H0)} is both necessary and sufficient. For every feasible $p_x$, the horizontal top lies entirely in the common level component and is therefore contained in the terrain. Its area is $\frac{h_i}{2}\bigl((b_r)_x+(t_r)_x-2x_i\bigr)$, which is independent of the particular feasible value of $p_x$. The component and carrier orders therefore give an $O(n)$ compatibility and optimization scan for the fixed balanced centre. Since there are $O(n)$ possible centres, all balanced-left--carrier configurations require $O(n^2)$ time. The balanced-right case follows by reflection.

We next consider two balanced sides. Let their midpoint contacts be $v_i=(x_i,y_i)$, and $v_j=(x_j,y_j)$, with $y_i=y_j=h_i/2$, and $x_i<x_j$. If their top endpoints are $(x_L,h_i)$ and $(x_R,h_i)$, then the corresponding base endpoints are $(2x_i-x_L,0)$ and $(2x_j-x_R,0)$. Thus the corresponding area $\area(Q)=h_i(x_j-x_i)$, depends only on the two balanced centres and not on the particular compatible top endpoints.

Containment of the horizontal top requires the two top endpoints to belong to the same component of $\{x:T(x)\ge h_i\}$. This condition is also sufficient for containment of the top segment, because the entire horizontal segment between the endpoints then lies at height at most $T(x)$.

For a level component $I_\nu$, let $J^L_\nu=[\lambda^L_\nu,U^L_\nu]$, and $J^R_\nu=[\lambda^R_\nu,U^R_\nu]$ be the feasible balanced-endpoint intervals of the two sides, whenever they are nonempty. Compatibility within $I_\nu$ is equivalent to the existence of $x_L\in J^L_\nu$, and $x_R\in J^R_\nu$ satisfying
\[
        x_L<x_R
        \qquad\text{and}\qquad
        x_R-x_L<2(x_j-x_i).                                  \tag{H1}
\]
Let $\delta_{ij}=2(x_j-x_i)$. As $x_L$ ranges over $J^L_\nu$, the values of $x_R$ permitted by \textup{(H1)} form $(\lambda^L_\nu,\;U^L_\nu+\delta_{ij})$. Hence the component is compatible exactly when $J^R_\nu\cap(\lambda^L_\nu,\;U^L_\nu+\delta_{ij})\ne\varnothing$. The two component-indexed interval lists are already ordered by component. A simultaneous scan therefore tests only pairs having the same component index, with constant work for each such pair. Its running time is linear in the total number of intervals, and endpoints belonging to different terrain valleys are never combined.

Equality in the first inequality of~\textup{(H1)} collapses the horizontal top, while equality in the second collapses the base width. Both give triangular boundary configurations and are covered globally by $\Delta^*$.

Under general position, no three terrain vertices share a height. Hence each balanced centre has at most one distinct height-matched partner, so there are only $O(n)$ such pairs. Scanning the $O(n)$ interval data for each pair gives $O(n^2)$ total time, while processing one centre or matched pair at a time uses $O(n)$ space.

The balanced--exposed, balanced--carrier, and double-balanced cases exhaust the horizontal configurations not already handled by the carrier--carrier, carrier--exposed, and doubly-exposed analyses. This proves the lemma.~\hfill$\qed$
\end{proof}

The preceding lemmas cover all cases of Corollary~\ref{cor:algorithmic-side-classification}. Combining their bounds with the grounding and classification results and the computation
of $\Delta^*$ yields the final theorem.



\begin{theorem}\label{thm:exact}
For an $n$-vertex $1.5$D terrain in general position, the maximum area over the closure of its convex quadrilaterals can be computed deterministically in $O(n^2\log n)$ time and $O(n)$ space in the algebraic real-RAM. If a nondegenerate quadrilateral attains the optimum value, the algorithm returns one; otherwise the algorithm returns $\Delta^*$. Equal-area ties are
resolved in favor of a quadrilateral.
\end{theorem}

\begin{proof}
We first prove correctness. If no nondegenerate quadrilateral attains the optimum over the closed feasible family, then the optimum is represented by a feasible terrain triangle. Hence the maximum-area terrain triangle $\Delta^*$, computed by Cabello et al.~\cite{CabelloDDM25}, attains the optimum value.

Suppose instead that a nondegenerate quadrilateral attains the optimum. By Lemma~\ref{lem:grounding}\textup{(ii)}, some optimum is grounded. Lemma~\ref{lem:twopoint}, Lemma~\ref{lem:active-forward}, and Corollary~\ref{cor:algorithmic-side-classification} then classify its two base-incident sides as carrier-representable, balanced, or incident to an
exposed upper corner.

For a nonzero-slope top, reflection reduces the analysis to positive slope. If neither upper corner is exposed, Observation~\ref{obs:boundary-contact} gives an upper-chain contact in the relative interior of the top edge. When both side supports are carrier-representable, the configuration is covered by Lemmas~\ref{lem:terminal-order}, \ref{lem:top-frontier},
and~\ref{lem:carrier-branch}. A balanced--carrier optimum not already covered by the carrier--carrier branch is represented by Lemma~\ref{lem:balanced-top-carriers} and optimized by
Lemma~\ref{lem:balanced-branch}.

If exactly one upper corner is exposed, the positive-slope possibilities are right-exposed with a carrier or balanced left side, and left-exposed with a right carrier. Lemma~\ref{lem:moving-corner-envelope} handles all three, and reflection gives the negative-slope cases. If both upper corners are exposed, Lemma~\ref{lem:two-exposed} applies.

For a horizontal top, the carrier--carrier case appears as the zero-slope boundary of Lemma~\ref{lem:carrier-branch}. Carrier--exposed and doubly-exposed configurations are covered by
Lemmas~\ref{lem:moving-corner-envelope} and~\ref{lem:two-exposed}, respectively, while Lemma~\ref{lem:quadratic-horizontal} handles the remaining balanced--exposed, balanced--carrier, and double-balanced cases. Thus every nondegenerate optimum is represented by one of the algorithmic branches.

Within each branch, the maximum-level preference test retains a nondegenerate representative whenever both degenerate and nondegenerate candidates attain the branch optimum. After all branches are processed, the best quadrilateral candidate is compared with $\Delta^*$, with equality resolved in favor of the nondegenerate quadrilateral.

We next bound the running time. Lemmas~\ref{lem:carrier-branch} and~\ref{lem:balanced-branch} give $O(n^2\log n)$ time for the carrier--carrier and balanced--carrier branches. By
Lemma~\ref{lem:moving-corner-envelope}, each fixed one-exposed carrier or balanced representation is processed in $O(n)$ time; over $O(n)$ such representations this contributes $O(n^2)$ time, including the horizontal carrier--exposed boundary cases. Lemma~\ref{lem:two-exposed} handles the doubly-exposed branch in $O(n^2)$ time, and
Lemma~\ref{lem:quadratic-horizontal} handles the remaining horizontal configurations in $O(n^2)$ additional time. Finally, $\Delta^*$ is computed in $O(n\log n)$ time. Hence the total running time is $O(n^2\log n)$.

Each branch uses $O(n)$ working space, and the branches are processed sequentially while retaining only the best constant-size candidate found so far. The hereditary segment-tree implementation of Cabello et al.~\cite[Section~4]{CabelloDDM25} used to compute $\Delta^*$ can likewise be processed level by level in $O(n)$ working space. Therefore the complete
algorithm uses $O(n)$ space.~\hfill$\qed$
\end{proof}

\section[A 1/2-Approximation via an Axis-Parallel Rectangle]{A $\frac{1}{2}$-Approximation via an Axis-Parallel Rectangle}\label{sec:approx}

The exact algorithm above uses the full structural classification of optimal
quadrilaterals. A much simpler argument gives a constant-factor
approximation. Boland and Urrutia~\cite{BolandU01a} showed how to compute a
maximum-area axis-parallel rectangle in a simple polygon in
$O(n\log n)$ time. We prove that, in a terrain, the area of this rectangle
is always at least one half of the optimum value of our quadrilateral
problem. The factor cannot be improved in general.

The key geometric statement does not depend on the terrain structure.
Consider any grounded convex polygon whose two base angles are at most
$90^\circ$. Its horizontal sections are nested as the height increases.
If $w(y)$ denotes the width of the section at height $y$, then the section
at that height determines a grounded axis-parallel rectangle of area
$y\,w(y)$. Concavity of the width function is enough to show that some such
rectangle captures at least half of the polygon area. We formulate this
fact separately because it may be useful beyond the terrain setting.

\begin{lemma}\label{lem:section-rectangle}
Let $Q$ be a compact convex polygon with a side on the horizontal line $\Pi$, and suppose that its two base angles are at most $90^\circ$. Then $Q$ contains an axis-parallel rectangle of area at least $\frac12\area(Q)$.
\end{lemma}

\begin{proof}
Let $H$ be the height of $Q$. For $y\in[0,H]$, let $I(y)=[L(y),R(y)]$ be the horizontal section of $Q$ at height $y$, and let $w(y)=R(y)-L(y)$ be the horizontal width of $Q$ at height $y$.

Convexity of $Q$ implies that every $I(y)$ is an interval. Since the two base angles are at most $90^\circ$, the left endpoint $L(y)$ is nondecreasing in $y$, while the right endpoint $R(y)$ is nonincreasing. Hence the horizontal sections are nested: $y_1\le y_2  \Longrightarrow  I(y_2)\subseteq I(y_1)$. Consequently, for every $t\in[0,H]$, $\mathcal R_t=I(t)\times[0,t]$ is an axis-parallel rectangle contained in $Q$. Its width is $w(t)$ and its height is $t$, so $\area(\mathcal R_t)=t\,w(t)$.

The width function $w$ is concave. Indeed, for $y_1,y_2\in[0,H]$ and $\lambda\in[0,1]$, convexity of $Q$ implies that the corresponding convex combinations of the left and right section
endpoints belong to $Q$ at height $y=(1-\lambda)y_1+\lambda y_2$. Therefore $w(y) \ge (1-\lambda)w(y_1)+\lambda w(y_2)$, which proves concavity.

We now define $M=\max_{0\le y\le H} y\,w(y)$. Thus $M$ is the largest area among the grounded axis-parallel rectangles $\mathcal R_y=I(y)\times[0,y]$ obtained from the horizontal sections of $Q$. Let $t\in[0,H]$ be a height at which this maximum is attained, so that $M=tw(t)$. If $\area(Q)=0$, the claim is immediate. We therefore assume $\area(Q)>0$. Then $M>0$, and consequently $t>0$.

We first consider the case $t<H$, in which the maximizing rectangle does not reach the topmost level of $Q$. By the interior optimality of $y\,w(y)$ at $t$ and the concavity of $w$, there exists a supergradient $s$ of $w$ at $t$ satisfying $w(t)+ts=0$. Geometrically, $s$ is the slope of a supporting line to the graph of $w$ at the point $(t,w(t))$. The supporting-line inequality gives
\[
        w(y)\le w(t)+s(y-t)=w(t)\left(2-\frac{y}{t}\right) \qquad\text{for all }y\in[0,H].
\]
Since $w(y)\ge0$, this upper bound implies $H\le2t$. Thus we obtain
\[
        \area(Q) \le  \int_0^{2t} w(t)\left(2-\frac{y}{t}\right)dy =2tw(t)  =2M.
\]

We next consider the endpoint case $t=H$, in which the largest grounded rectangle is determined by the topmost section. Let $m_H$ denote the slope of the final linear piece of the concave width function $w$ immediately below $H$. Equivalently,
\[
        m_H  = \lim_{\substack{y\to H\\ y<H}}\frac{w(H)-w(y)}{H-y}.
\]
Since $H$ maximizes $y\,w(y)$ over $[0,H]$, one-sided optimality gives $w(H)+Hm_H\ge0$, and therefore $m_H\ge-\frac{w(H)}{H}$.

By concavity, the supporting line through $(H,w(H))$ with slope $m_H$ lies above the graph of $w$, so
\[
        w(y)\le w(H)+m_H(y-H)\qquad\text{for }0\le y\le H.
\]
Since $y-H\le0$ and $m_H\ge-\frac{w(H)}{H}$, we obtain $w(y)\le w(H)\left(2-\frac{y}{H}\right)$. Hence
\[
        \area(Q)=\int_0^H w(y)dy \le \int_0^H w(H)\left(2-\frac{y}{H}\right)dy =\frac32H w(H) =\frac32M \le2M.
\]

Thus, in both cases, $M\ge\frac12\area(Q)$. By definition, $M=t\,w(t)$ is the area of the grounded rectangle $\mathcal R_t=I(t)\times[0,t]\subseteq Q$. Therefore $Q$ contains an axis-parallel rectangle of area at least $\frac12\area(Q)$.~\hfill$\qed$
\end{proof}


We now apply Lemma~\ref{lem:section-rectangle} to an optimum of the terrain problem. By Lemma~\ref{lem:grounding}\textup{(ii)}, an optimum can be chosen grounded; Lemma~\ref{lem:grounding}\textup{(i)}, with the same base-endpoint perturbation in the triangular case, gives the required base-angle condition. Hence it contains an axis-parallel rectangle of area at least $\frac12\area(Q^*)$. Since $\mathcal R^*$ is a largest such rectangle in
$\mathfrak T$, the same bound holds for $\mathcal R^*$. A triangular terrain shows that the factor $\frac12$ is tight.

\begin{theorem}\label{thm:apprx}
Let $Q^*$ be an optimum and $\mathcal R^*$ a largest axis-parallel rectangle in $\mathfrak T$. Then $\area(\mathcal R^*)\ge \tfrac12\area(Q^*)$. Moreover, $\mathcal R^*$ is computable in $O(n\log n)$ time, and the factor $\frac12$ is tight.
\end{theorem}

\begin{proof}
By Lemma~\ref{lem:grounding}\textup{(ii)}, we may take $Q^*$ to be grounded. Thus $Q^*$ is an optimum whose base side lies on $\Pi$. If $Q^*$ is a nondegenerate quadrilateral, Lemma~\ref{lem:grounding}\textup{(i)} implies that its two base angles are at most $90^\circ$. If $Q^*$ is triangular, the same conclusion follows from the base-endpoint perturbation used in the proof of Lemma~\ref{lem:grounding}\textup{(i)}, with the two upper vertices identified. Indeed, a base angle larger than $90^\circ$ would allow the corresponding base endpoint to move outward while preserving feasibility and increasing the area, contradicting optimality.

Lemma~\ref{lem:section-rectangle} therefore gives an axis-parallel rectangle $\mathcal R\subseteq Q^*$ whose area satisfies $\area(\mathcal R)\ge\frac12\area(Q^*)$. Here $\mathcal R$ is the rectangle guaranteed inside the particular grounded optimum $Q^*$; it need not itself be the largest rectangle in the terrain. Since $Q^*\subseteq\mathfrak T$, the rectangle $\mathcal R$ is feasible in the terrain. By definition, $\mathcal R^*$ is the maximum-area axis-parallel rectangle contained in $\mathfrak T$. Hence $\area(\mathcal R^*)\ge\area(\mathcal R)\ge\frac12\area(Q^*)$, where $\mathcal R^*$ is computed in $O(n\log n)$ time~\cite{BolandU01a}.

We finally establish tightness. We consider the case in which $\mathfrak T$ itself is a triangle with horizontal base of width $B$ and height $H$. At height $y$, its horizontal section has width $w(y)=B\left(1-\frac{y}{H}\right)$. For an axis-parallel rectangle whose lower and upper sides lie at heights $y_0$ and $y_1$, respectively, nestedness of the horizontal sections gives $\area(\mathcal R)\le(y_1-y_0)w(y_1)\le y_1w(y_1)$. Conversely, for every $y\in[0,H]$, the grounded rectangle $I(y)\times[0,y]$ is contained in the triangle. Therefore the largest rectangle has area
\[
        \max_{0\le y\le H} yB\left(1-\frac{y}{H}\right)=\frac{BH}{4},
\]
with the maximum attained at height $y=H/2$.

The area of the triangular terrain, and hence the optimum value under the
closed convention for this example, is $\area(Q^*)=\frac{BH}{2}$. 
Consequently,
\[
        \frac{\area(\mathcal R^*)}{\area(Q^*)}=\frac{BH/4}{BH/2}=\frac12.
\]
Thus the approximation factor $\frac12$ is tight.~\hfill$\qed$
\end{proof}

\bibliographystyle{splncs04}
\bibliography{caldam}

\end{document}